\documentclass[11pt]{article}
\usepackage{graphicx} %
\usepackage{pdfpages}
\usepackage{bm}
\usepackage{amsmath,amssymb}
\usepackage{macros}
\usepackage{quantikz}
\usepackage{cleveref}
\crefname{figure}{figure}{figures}
\Crefname{figure}{Figure}{Figures}
\usepackage{physics}
\usepackage{algorithmic}
\usepackage{bbold}
\usepackage{multirow}
\usepackage{array}
\usepackage{makecell}
\usepackage{comment}
\usepackage{subcaption}
\usepackage{mathpazo}
\usepackage{authblk}
\usepackage{fullpage}
\usepackage{booktabs}
\usepackage[table]{xcolor}
\usepackage{makecell}
\usepackage{tabularx}

\newcommand{\numrepetitionszero}{31}
\newcommand{\totalqubitszero}{525}
\newcommand{\depthqspzero}{1.53}
\newcommand{\depthstatezero}{9.01}
\newcommand{\gateqspzero}{0.07}
\newcommand{\gatestatezero}{1.70}
\newcommand{\totaldepthzero}{0.33}
\newcommand{\totalgateszero}{0.05}
\newcommand{\classicalflopszero}{0.51}

\newcommand{\numrepetitionsone}{89}
\newcommand{\totalqubitsone}{720}
\newcommand{\depthqspone}{2.19}
\newcommand{\depthstateone}{12.9}
\newcommand{\gateqspone}{0.14}
\newcommand{\gatestateone}{3.38}
\newcommand{\totaldepthone}{1.33}
\newcommand{\totalgatesone}{0.31}
\newcommand{\classicalflopsone}{115}

\newcommand{\numrepetitionstwo}{201}
\newcommand{\totalqubitstwo}{900}
\newcommand{\depthqsptwo}{2.88}
\newcommand{\depthstatetwo}{16.9}
\newcommand{\gateqsptwo}{0.23}
\newcommand{\gatestatetwo}{5.55}
\newcommand{\totaldepthtwo}{3.97}
\newcommand{\totalgatestwo}{1.16}
\newcommand{\classicalflopstwo}{6611}

\newcommand{\numrepetitionsfive}{393}
\newcommand{\totalqubitsfive}{1110}
\newcommand{\depthqspfive}{3.58}
\newcommand{\depthstatefive}{21.1}
\newcommand{\gateqspfive}{0.35}
\newcommand{\gatestatefive}{8.67}
\newcommand{\totaldepthfive}{9.70}
\newcommand{\totalgatesfive}{3.54}
\newcommand{\classicalflopsfive}{1.6$\times 10^5$}

\renewcommand{\>}{\rangle}
\newcommand{\<}{\langle}
\newcommand{\SC}{\mathcal{S}}
\newcommand{\KC}{\mathcal{K}}

\title{End-to-end quantum algorithms for tensor problems}
\author{Enrico Fontana\thanks{\texttt{enrico.fontana@jpmchase.com}}}
\author{Sivaprasad Omanakuttan}
\author{Junhyung Lyle Kim}
\author{Joseph Sullivan}
\author{Michael Perlin}
\author{Ruslan Shaydulin}
\author{Shouvanik Chakrabarti\thanks{\texttt{shouvanik.chakrabarti@jpmchase.com}}}
\affil{Global Technology Applied Research, JPMorganChase, New York, NY 10001, USA}
\date{October 2025}

\begin{document}

\maketitle
\begin{abstract}
    We present a comprehensive end-to-end quantum algorithm for tensor problems, including tensor PCA and planted kXOR, that achieves potential superquadratic quantum speedups over classical methods. We build upon prior works by Hastings~(\textit{Quantum}, 2020) and Schmidhuber~\textit{et al.}~(\textit{Phys.~Rev.~X.}, 2025), we address key limitations by introducing a native qubit-based encoding for the Kikuchi method, enabling explicit quantum circuit constructions and non-asymptotic resource estimation. Our approach substantially reduces constant overheads through a novel guiding state preparation technique as well as circuit optimizations, reducing the threshold for a quantum advantage. We further extend the algorithmic framework to support recovery in sparse tensor PCA and tensor completion, and generalize detection to asymmetric tensors, demonstrating that the quantum advantage persists in these broader settings. Detailed resource estimates show that 900 logical qubits, $\sim 10^{15}$ gates and $\sim 10^{12}$ gate depth suffice for a problem that classically requires $\sim 10^{23}$ FLOPs. The gate count and depth for the same problem without the improvements presented in this paper would be at least $10^{19}$ and $10^{18}$ respectively. These advances position tensor problems as a candidate for quantum advantage whose resource requirements benefit significantly from algorithmic and compilation improvements; the magnitude of the improvements suggest that further enhancements are possible, which would make the algorithm viable for upcoming fault-tolerant quantum hardware.
\end{abstract}

\section{Introduction}
\subsection{Motivation}
Speedups for quantum algorithms over their best classical counterparts, where present, generally fall into two camps: exponential and quadratic, the latter typically stemming from amplitude amplification. Surprisingly, only a handful of problems are known where there exists a quantum algorithm that has superquadratic (but still polynomial) speedup. One such problem is that of spiked tensor PCA and planted kXOR. The problem is stated as follows: given observations of a symmetric $k$-order tensor $T \in (\mathbb{R}^n)^{\otimes k}$ that is promised to be of the form 
$$T = \lambda\vec{z}^{\otimes k} + N,$$
with $\vec{z} \in \mathbb{R}^n$ the spike or planted vector, $N$ a symmetric noise tensor and $\lambda \in \mathbb{R}^+$ the signal-to-noise ratio (SNR), either
\begin{itemize}
    \item with the promise that either $\lambda = 0$ or $\lambda \ge \lambda_{thr}$ , determine which case is applicable to the input tensor \textbf{(detection)}, or
    \item extract the planted vector $\vec{z}$ \textbf{(recovery)}.
\end{itemize}
The first quantum algorithm was developed by Hastings~\cite{hastings2020classical} for tensor PCA (all entries observed) with the assumption of Gaussian $N$ and generic $\vec{z} \in \mathbb{R}^n$, for both detection and recovery. Following this, Schmidhuber \textit{et al.}~\cite{schmidhuber2024quartic} presented a quantum algorithm for detection for planted kXOR, which is equivalent to having only $m \ll n^k$ observations of a symmetric tensor $T \in (\{\pm1\}^n)^{\otimes k}$ obtained from a signed spike $\vec{z} \in \{\pm 1\}^n$ and where the noise is given by flipping the sign of the entries with some probability. The two algorithms share the following properties:
\begin{itemize}
    \item They utilize the \textbf{Kikuchi method}~\cite{wein2019kikuchi}: the solution is given by finding an eigenvector of high energy of a Hamiltonian called the Kikuchi matrix, constructed from the tensor;
    \item The speedup is proven over a classical algorithm that simply uses the power method on the Kikuchi matrix, believed (in some settings) to be state-of-the-art;
    \item The speedup is asymptotically \textbf{quartic}, and conceptually consists of two concatenated quadratic speedups: one Grover-like from amplitude amplification, and another that comes from preparing an appropriate \textbf{guiding state} with improved overlap with the high-energy subspace.
\end{itemize}
These characteristics make these methods a particularly promising avenue for practical quantum advantage. A superquadratic speedup is believed to be necessary due to the overheads of quantum error correction~\cite{babbush2021focus}, however the vast majority of problems with exponential advantage are extremely structured and therefore of limited value. In contrast, these algorithms appear to promise a superquadratic speedup for a relatively unstructured problem. Furthermore tensors as data structures are ubiquitous, as they are the most natural way to store multidimensional datasets. The problem of PCA is also a very natural one.

There are considerable problems with the practical implementation of these quantum algorithms. A major issue is that the quartic advantage is only asymptotic: the actual runtime of the algorithms includes very large overheads that make their execution all but impractical and that are not easily removed using the existing approaches. The algorithms are also framed in terms of qudits, which are a relatively unwieldy setting. Also, the scaling of the runtime with $m$ means that the setting of dense tensor PCA ($m \sim n^k$) is unfavored as the number of entries grows rapidly with the tensor size $n$. The algorithm for planted kXOR is potentially more practical as it only requires $m \sim n^{k/2}$. Finally, both algorithms suffer from the key limitation of requiring \textbf{symmetric} tensors $T$, whereas many practical applications have data that does not have any such symmetry.

In this work, we present the following contributions:
\begin{itemize}
    \item By using a different encoding, we construct circuits that are natively \textbf{qubit-based};
    \item By modifying the guiding state preparation, we greatly reduce the overheads.
    \item We build \textbf{explicit circuits}, allowing us to do non-asymptotic resource estimation.
    \item We provide a quantum algorithm for \textbf{recovery} that applies for the case of sparse tensor PCA and tensor completion.
    \item We extend the theory for tensor PCA (detection) to the setting of \textbf{asymmetric tensors}, showing that the superquadratic advantage persists.
\end{itemize}
Our work therefore uncovers a potential superquadratic quantum advantage for a variety of practical, unstructured problems, and provides explicit circuits allowing for non-asymptotic resource estimation. These results suggest that tensor problems such as tensor PCA and completion should be regarded as one of the most promising candidates for quantum computing applications.

Here we cover broadly the two classes of problems that the Kikuchi method can solve, kXOR and tensor problems. In fact the planted kXOR problem can be seen as an instance of tensor PCA for a spiked tensor. However, in this paper we will use the language of kXOR problems as the algorithms are more easily discussed in this context.

\subsection{Planted kXOR problems} 
\label{sec:planted-kxor}
A kXOR problem is a set of clauses of binary variables $x_1, ..., x_n$, $x_i \in \{0,1\}$, such that each clause contains exactly $k$ distinct variables and each clause is assigned a binary value. More precisely, a kXOR problem is defined as a set of tuples $\{(C_1, b_1), ..., (C_m, b_m)\}$ of size $m$ where $C_i$ is the set of labels for the variables composing the clause, and $b_i\in \{0,1\}$ is the clause's binary value assignment:
\begin{equation}
    (C_i, b_i) \implies \bigoplus_{j \in C_i} x_j = b_i
\end{equation}
The Max-kXOR problem asks to determine the largest fraction of clauses $f_{\max}$ that can be simultaneously satisfied by an assignment of variables $\vec{x} = (x_1, ..., x_n)$.
As kXOR is a form of k-SAT and more generally a Constrained Satisfaction Problem (CSP), kXOR problems like Max-kXOR are prototypical NP-hard problems, and by the Exponential Time Hypothesis~\cite{impagliazzo2001complexity} there is no algorithm that can solve an arbitrary problem of this class in time $\exp(o(n))$. Though polynomial-time approximation schemes (PTASs) have been shown to exist for dense Max-kCSPs ($m \sim n^k$)~\cite{arora1995polynomial}, producing an assignment that satisfies a fraction of the clauses larger than $rf_{\max}$ with $r<1$ a constant is hard in the worst case when $m = O(n^{k-1})$~\cite{fotakis2015sub}.

However, one can choose the labels and assignments from some distribution and get more tractable problems on average. For instance, there exist sub-exponential time tight refutation algorithms for random Max-kCSPs for $m = \omega(n)$ and certification algorithms for $m = \tilde{\Omega}(n^{k/2})$~\cite{allen2015refute, barak2016noisy,raghavendra2017strongly}, which have been extended to the semirandom case (arbitrary clauses and random assignments)~\cite{guruswami2022algorithms}.

The kXOR setting that will be the main focus of this work, planted kXOR, is one such specialization of kXOR that admits PTASs in some settings. We start by considering a planted kXOR instance, obtained by sampling a planted solution $\vec{z} \in \{0,1\}^n$, and randomly choosing $m$ clauses satisfied by it. Then each clause is corrupted by flipping its assignment bit with some probability $\eta$. Sometimes this is also written in terms of the \textit{planted advantage} $\rho := 1-2\eta$.
We can then consider the planted kXOR decision problem, which is the problem that asks to decide if a given kXOR instance is either planted with error probability $\eta$, or random (uniform). 
An algorithm for planted kXOR should distinguish between the two cases. Alternatively one can consider the inference problem of recovering the planted solution $\vec{z}$.
For either problems there are no known efficient algorithms for $m/n \ll n^{k/2-1}$, and it is conjectured that in that regime the problem is intractable~\cite{alekhnovich2001lower, schoenebeck2008linear, odonnell2014goldreich, mori2016lower, kothari2017sum, daniely2016complexity, raghavendra2017strongly}. 
More precisely, the runtime is expected to be $\exp{(O(n^{(1+\delta)/2})}$ for $m/n = \tilde{\Omega}(n^{(k/2-1)(1-\delta)})$, $0<\delta<1$~\cite{chen2024algorithms}. 
In fact, the hardness of planted kXOR in this regime has been used as a hardness assumption in cryptography~\cite{alekhnovich2003more, applebaum2016cryptographic}, where the problem is known as sparse learning parity with noise (sparse LPN)~\cite{grigorescu2011noise}.
However, for larger clause densities $\Omega(n^{k/2-1})$ the problem becomes easy, i.e. polynomial time algorithms exist for distinguishing~\cite{applebaum2016cryptographic, allen2015refute} and for inference~\cite{barak2016noisy}. This threshold emerges intuitively from an application of the ``birthday paradox'': given $m = \Omega(n^{k/2})$ clauses we expect a number of clauses $\Omega(n)$ to differ by only one variable, and thus we reduce the problem to 2XOR which can be solved efficiently~\cite{applebaum2016cryptographic}.

Historically SoS has yielded the best algorithms for kXOR problems. The degree-$\ell$ SoS method runs in time $n^{O(\ell)}$ and can solve kXOR problems like refutation and planted inference, however, it has been shown to fail for strong refutation at densities $m/n \ll n^{k/2-1}$. Additionally, with high probability the SoS method at degree $\ell = O(n^\delta)$ cannot refute if $m/n \ll n^{(1-\delta)(k/2-1)}$. More precisely, Ref.~\cite{raghavendra2017strongly} showed that degree-$\ell$ SoS can solve the inference problem if $m/n = \tilde{\Omega}(\left(n/\ell\right)^{k/2-1})$. Conversely, SoS cannot solve the problem unless $m/n \gg (\frac{n}{\ell\log(n/\ell)})^{k/2-1}$~\cite{schoenebeck2008linear}, meaning that the computational threshold for SoS is tight up to polylog factors around a clause density of $(\frac{n}{\ell})^{k/2-1}$.
This agrees with the conjectured hardness regime discussed above. Therefore we do not expect generic algorithms for either strong refutation or planted inference to be efficient at much lower clause densities than the SoS thresholds~\cite{schmidhuber2024quartic}. Note that the informational threshold for planted inference is simply $m/n \sim 1$, therefore for high $k$ there is actually a considerable informational-computational gap.

We note that there also exists a dependence on the error rate $\eta$ in the clause density required for success and the runtime. However, there is much less focus on the computational regimes induced by the error rate in the literature. In the (conjectured) exponential-time regime of $m/n \ll n^{k/2-1}$, the runtime gets elevated to $\exp(\tilde{O}(\eta n))$~\cite{chen2024algorithms}. Other algorithms work only when restricted to the low-noise regime~\cite{dao2024lossy}.

\subsection{Tensor problems}
\paragraph{Tensor PCA}
Tensor PCA is a generalization of ordinary PCA to tensors. Specifically, suppose we decompose a $k$-order tensor $T \in \mathbb{R}^{n_1\otimes ...\otimes n_k}$ into its minimal decomposition
$$T = \sum_{r=1}^R \lambda_r u_{1}^{(r)} \otimes ...\otimes u_k^{(r)}$$
where $R$ is the tensor rank, $\|u_i^{(r)}\|=1$ and $\lambda_1 \ge ... \ge \lambda_R$ are the tensor singular values. Then tensor PCA is the task of returning the top rank-1 component $\lambda_1 u_{1}^{(1)} \otimes ...\otimes u_k^{(1)}$.
However, calculating tensor rank itself is much harder to determine than normal matrix (in fact, it is NP-hard~\cite{haastad1989tensor}) and so we expect optimal tensor decomposition to be an exceedingly hard problem in general. A more direct route to tensor PCA is optimal rank-1 approximation, which is the minimization of the quadratic form 
$$\|T-\widehat{T}\|^2_F = \sum_{i_1...i_k} (T_{i_1...i_k}-\widehat{T}_{i_1...i_k})^2$$
over rank-1 tensors $\widehat{T}$. This leads to algorithms like higher-order power iteration and generalizes to higher-order SVD (HOSVD)~\cite{de2000best} / Tucker decomposition~\cite{tucker1966some}, or via the Canonical Polyadic / CANDECOMP-PARAFAC (CP) decomposition~\cite{hitchcock1927expression, carroll1970analysis, harshman1970foundations}. While various heuristic algorithms for these decompositions exist~\cite{kolda2009tensor}, overall optimal rank-1 approximation is still NP-hard~\cite{hillar2013most}.

Another interesting approach consists in defining a prior for the tensor. 
Consider the spiked tensor model~\cite{richard2014statistical}, which defines a symmetric tensor $T \in \mathbb{R}^{n^{\otimes k}}$ as a perturbation of a rank-1 component:
\begin{equation}
    T = \lambda \vec{z}^{\otimes k} + N
\end{equation}
where $\lambda \in \mathbb{R}$ is the signal-to-noise ratio (SNR), $\vec{z} \in \mathbb{R}^n$ is an arbitrary vector (the ``spike'')(usually normalized as $\|\vec{z}\|=\sqrt{n}$) and $N$ is some symmetric noise tensor (usually standard normal). In a sense then $\vec{z}^{\otimes k}$ is the principal component of the tensor, obfuscated by the noise and the task of recovering this component (with some error) may be called tensor PCA. 

Under a prior for $\vec{z}$ and $N$, this is the statistical model for tensor PCA that was introduced by Montanari and Richard~\cite{richard2014statistical} and shown to admit tractable estimators in certain regimes, specifically those based on tensor unfolding and tensor power iteration. Briefly, tensor unfolding flattens the $k$-order tensor into a matrix of size $n^q \times n^{k-q}$ for some $q$, and performs PCA on that matrix. Tensor power iteration involves applying the tensor to $k-1$ copies of a random vector, normalizing the resulting vector, and repeating until convergence. It was shown that tensor unfolding succeeds down to a SNR ratio of roughly $\sim n^{k/4}$, while tensor power iteration works above $\sim n^{(k-1)/2}$, meaning that tensor unfolding is more powerful than power iteration. AMP, while being a refined version of power iteration, also fails below $\lambda \sim n^{(k-1)/2}$.
The authors also prove the existence of an informational threshold at $\lambda \sim \sqrt{n}$ under which the spiked tensor is provably indistinguishable from a random tensor, but above which maximum likelihood estimation (returning the optimal rank-1 approximation) succeeds. Again, this likely is not an efficient procedure, therefore we have an instance of information-computation gap. As with other similar gaps, the failure of local algorithms to go below $\sim n^{k/4}$ is linked to the optimization landscape becoming extremely rugged below that threshold~\cite{arous2019landscape}.

The problem of tensor PCA can be simply related to Max-kXOR. Consider an $k$-order symmetric signed tensor $T$, i.e. with entries in $\{\pm 1\}$. Every such tensor has a corresponding dense kXOR problem with $m={n\choose k}$ clauses, each clause corresponding to an entry of the tensor with distinct indices via $2b_{i_1...i_k} = 1-T_{i_1,,,i_k}$; the converse direction is also true. The planted kXOR then corresponds to a spiked tensor inference, where the noise tensor is $-2\vec{z}^{\otimes k}\odot B$ with $B$ a symmetric tensor of binary random variables. In analogy with planted kXOR we can also give a decision problem version of tensor PCA, also called detection problem, when one is tasked to distinguish between a spiked tensor as above and the completely random case $T = N$. However, directly relating to kXOR is only possible for symmetric tensors.

\paragraph{Tensor completion}
In tensor completion, we are presented with $m$ observations from an $k$-order tensor of dimension $n$, where $m \ll n^k$, and are asked to complete the missing entries. In other words, given a sparse tensor of observations $\mathrm{T}$ sampled from a dense tensor $\mathrm{T}^*$, we are asked to produce some representation of a dense tensor $\mathrm{T^\#}$ such that in some tensor norm $\|\mathrm{T}^\# - \mathrm{T}^*\| \le \epsilon$. In the literature, this is usually done by low-rank decomposition of $\mathrm{T}$, for instance using the CP decomposition, and then taking $\mathrm{T^\#}$ to be the low-rank components~\cite{acar2011scalable,song2019tensor}. The intuition is that the principal components of $\mathrm{T}$ approximate well the principal components of $\mathrm{T}^*$. On this front, one important result is that of Yuan and Zhang~\cite{yuan2017incoherent} that a noise-free rank-$r$ tensor can be fully reconstructed using nuclear norm minimization from $\tilde{\Omega}(n^{k/2})$ observations.

For the noisy case, Barak and Moitra~\cite{barak2016noisy} give a polynomial-time convex optimization algorithm based on SoS for the detection version of noisy tensor completion in the regime $m = \tilde{\Omega}(n^{k/2})$, that works not only for asymmetric tensors but also for low-rank. However, an explicit runtime for the algorithm is not given. Later Montanari and Sun~\cite{montanari2018spectral} gave an analogous result but using a tensor unfolding method that achieves a low polynomial runtime $O(n^{k+a})$ for a small constant $a$. The algorithm by Cai \textit{et al.}~\cite{cai2019nonconvex} is also noteworthy, as it achieves low-rank tensor completion in a noisy setting using a gradient-descent-based algorithm, provably in $O(n^k)$ iterations.

Due to the natural connection between tensor PCA, tensor completion, and kXOR, it is no surprise that there are strong commonalities between the results. For instance assuming a spiked tensor model for a symmetric, signed tensor, tensor completion is equivalent to a sparse planted kXOR problem. The identical informational-computational gap and thresholds for number of observations follow immediately. Similarly many algorithms are common, including convex optimization (SoS) based algorithms and spectral methods like unfolding. This also suggests that, despite the similar asymptotic thresholds, some algorithms may offer an advantage over others in specific regimes. For instance, from the results on kXOR we would expect the SoS and Kikuchi hierarchies to provide a range of algorithms suitable for low observation ratios and/or SNRs. To this end we note that there are few theoretical guarantees on the regimes of noise that allow sparse tensor completion beyond the signed, symmetric tensors, though a recent work by Bandeira \textit{et al.}~\cite{bandeira2025tensor} helps shed some light on the matter.

\subsection{The Kikuchi method and quantum algorithms}
The Kikuchi method has emerged as a strong alternative to the SoS hierarchy for planted kXOR and tensor PCA problems, allowing both detection and recovery at regimes competitive with SDP methods and tensor unfolding~\cite{wein2019kikuchi, guruswami2022algorithms}. In many cases, it represents the ``best of both worlds'': it offers a hierarchy of increasingly powerful algorithms from both the SNR and observation ratio perspective while retaining the simplicity of spectral methods. 

The central quantity of the method is the Kikuchi matrix, or Kikuchi Hessian, which was introduced by Kikuchi~\cite{kikuchi1951theory} as a generalization of the Bethe Hessian in statistical physics. To introduce this matrix, let us consider a generic kXOR instance. This is defined by $m$ indices $S \in {[n]\choose k}$ that we collect in the set $\SC$. By ${[n] \choose k}$ we mean the set of all possible combinations of $k$ elements from $[n]=\{1,...,n\}$, which are themselves sets. In the notation from before, the instance is $\{(S_1, b_1), ..., (S_m, b_m)\}$, however this time we will shift from binary to signed variables $\vec{x} \in \{\pm 1\}^n$, such that $b_m \in \{\pm 1\}$ and $(S, b)$ represents the clause $x_S = b$, where by $x_S$ we mean $\prod_{i\in S}x_i$.

The order-$\ell$ Kikuchi matrix $\KC_{\ell}$ is indexed by the elements of ${[n]\choose \ell}$, where typically $\ell = ck$ for $c$ a positive integer. If $T,U \in {[n]\choose \ell}$, the entry $[\KC_{\ell}]_{T,U}$ only depends on the symmetric difference of the multi-indices:
$$[\KC_{\ell}]_{T,U} = \begin{cases}
    b_{T\Delta U} &\text{if } T\Delta U \in \mathcal{S},\\
    0 &\text{otherwise.}
\end{cases}$$

Despite its simple structure, the Kikuchi matrix possesses a remarkable power for detecting planted solutions in kXOR problems, or spiked tensors for tensor PCA. In fact, Wein \textit{et al.}~\cite{wein2019kikuchi} showed that the large eigenvalues of the Kikuchi matrix encode information about the planted solution. Specifically, for a random kXOR instance or tensor, 
the spectral norm $\|\KC_\ell\|$ is bounded with high probability by a quantity that is $\tilde{O}(\ell)$. Vice versa, in the setting of signed symmetric tensor PCA, in the presence of a rank-1 component it can be shown that, with high probability, there exists an eigenvalue at least as large as $\sim\lambda n^{k/2} \ell^{k/2}$, where $\lambda$ is the SNR. In the setting of sparse kXOR with planted advantage $\rho$, we have a similar lower bound of $\sim \rho m n^{-k/2} \ell^{k/2}$ and therefore one can have detection with $m = \Omega(n^{k/2})$. \cite{wein2019kikuchi} also proves recovery for the tensor PCA case by the use of a \textit{voting matrix} constructed with the largest eigenvector of $\KC_\ell$.

Since the spectral method of \cite{wein2019kikuchi} is based on finding a large eigenvalue eigenvector of $\KC_\ell$, its complexity scales with its dimension which is $O(n^{\ell})$, and as such can become prohibitive quite quickly. Interestingly, quantum versions of the Kikuchi method have been developed with a runtime of $\tilde{O}(n^{\ell/4})$ and therefore show an asymptotically quadratic advantage over the classical counterpart. The first such version was developed by Hastings~\cite{hastings2020classical} independently from~\cite{wein2019kikuchi}, and it maps the problem of dense tensor PCA to a system of bosons. In fact, it implements a variation of the Kikuchi method with a slightly modified matrix that is nonetheless identical when restricted to a subspace. Like~\cite{wein2019kikuchi}, the method achieves detection and recovery, the recovery using the same voting matrix, however it exhibits an asymptotic quartic advantage. A following paper by Schmidhuber \textit{et al.}~\cite{schmidhuber2024quartic} helped bridge the gap by presenting a similar quantum method for detection in the sparse planted kXOR setting with quartic advantage, which performs a diagonalization of the Kikuchi matrix on the quantum device via quantum phase estimation. The authors used a qudit-based encoding which is similar to Hasting's bosonic approach and present a more granular analysis of the quantum resources involved. 

The quantum algorithms are notable because they are instances of the guided Hamiltonian problem~\cite{gharibian2022dequantizing}: namely, one constructs a \textit{guiding state} with an improved overlap with the ground state of interest (compared to a random state), and then performs a regular ground state preparation procedure. This framework is interesting because it has been shown to span all the complexity classes of quantum algorithms: for sparse or local Hamiltonians, when the guiding state has exponentially small overlap with the ground state the problem is QMA-hard, when the overlap is inverse-polynomial it is BQP-complete, and when it is constant it is in BPP ~\cite{gharibian2022dequantizing, cade2022improved}. 
In our setting, even though the speedup in this setting is only polynomial, there is reason to believe that the procedure may not be easily ``Groverized" (i.e. rendered only quadratic by improvements in classical algorithms). Specifically it appears conspicuous to the authors that the largest known separation between deterministic and quantum query complexity for a total boolean function is also quartic~\cite{ambainis2017separations}.

\subsection{Summary of contributions}

There remain considerable challenges if one wants to evaluate the viability of these algorithms for practical application on quantum hardware that can be expected to be available on a reasonable timescale. Specifically, we would like to perform a concrete resource analysis that enable us to determine a viable crossover threshold for quantum vs classical algorithms for this problem~\cite{beverland2022assessing,chakrabarti2021threshold,omanakuttan2025threshold}. The primary technical obstructions are a lack of concrete circuit constructions to enable such a resource analysis, overwhelming constant overheads that significantly increase the crossover threshold, and the limitations of current analysis that prevent the application of these methods to realistic problem settings. Our work aims to directly address these obstructions. We summarize our main contributions below in the context of limitations in prior work.

\paragraph{Concrete circuit constructions enable detailed resource estimates}
The quantum algorithms of Hastings~\cite{hastings2020classical} and Schmidhuber \textit{et al.}~\cite{schmidhuber2024quartic} provide only asymptotic estimates and do not explicitly construct the oracles and the state preparation unitaries. Though this could be done with relative ease from their prescriptions, there exists a more fundamental barrier in that these algorithms are framed in terms of qudits. Quantum systems with large local dimensions are an unwieldy setting both algorithmically and practically, as they must be converted into qubits to run on a binary quantum computer, which is expected to give additional overheads. To address this challenge, we construct a quantum algorithm that is natively binary by using a different, yet very natural, qubit-based encoding for the Kikuchi matrix that utilizes the hard-core boson to qubit mapping. The encoding is entirely equivalent to those considered in previous works, and therefore the existing analysis is applicable. Based on this encoding, we explicitly construct circuits for guiding state preparation and phase estimation. The details of the circuit construction are given in Section~\ref{sec:quantum_algo}.

\paragraph{Reduction of large constant overheads}
The quartic advantage previously claimed is an asymptotic speedup in the parameter $n$. The actual runtime of the algorithm in~\cite{schmidhuber2024quartic} takes the form $${O}(mn^{\ell/4})\cdot \ell^{O(\ell)} \cdot \mathrm{poly}\log n$$ where $m$ is the number of observed entries and $\ell$ is the order of the Kikuchi matrix. One sees that the factor $\ell^{O(\ell)}$ (to be precise, this factor is exactly $\ell^{\ell/2}$) becomes enormous even at small values of $\ell$. The runtime of the algorithm in \cite{hastings2020classical} also contains a similar factor of $e^{O(\ell)}$.
For the qubit-based embedding, we construct a novel guiding state preparation subroutine using a ``one-hot shuffling'' procedure, which reduces the runtime to 
$${O}(mn^{\ell/4})\cdot \boldsymbol{(\ell/k)^{\ell/2}} \cdot \mathrm{poly}\log n.$$
This is a saving of $k^{\ell/2}$ which can be a large improvement in practical regimes (For instance, our resource estimates will be performed with $\ell=16, k=4$, leading to an improvement of $\sim 6.5 \times 10^4$). The new embedding also allows us to leverage considerable parallelism in the algorithm, that further reduces the circuit depths. Details of these improvements are given in Section~\ref{sec:quantum_algo}. In Section~\ref{sec:resource_est} we present concrete resource estimates based on our improved construction. For spiked tensor PCA detection on a 4-order signed tensor with $n=100$ and $m \sim 5\times 10^{5}$ observed entries\footnote{This task classically takes $\sim 6\times 10^{23}$ FLOPs, which coincidentally is about 1 mole of FLOPs.}, we obtain a logical qubit count $\sim 900$, a total depth of $\sim 4\times10^{12}$ and a non-Clifford gate count of $\sim 10^{15}$. Without the algorithmic improvements and parallelization, the depth would be at least order $10^{18}$ and the gate count would exceed $10^{19}$. While our resource requirements are still out of reach for current devices, the magnitude of the improvement we achieved suggests that the algorithm may be open to further enhancements, which would make it viable for upcoming fault-tolerant systems.

\paragraph{Algorithmic Extension: Recovery for sparse tensor PCA} There is currently no proof of recovery for planted kXOR or sparse tensor PCA in existing work. In the quantum setting recovery is only proved by Hastings~\cite{hastings2020classical} for the dense case. Conversely, the proofs of~\cite{wein2019kikuchi} cannot be directly adapted since the quantum algorithms cannot guarantee that the top-most eigenvector of $\KC$ is returned, only one of high energy. This is a significant barrier to practical applications of the quantum Kikuchi methods, where full recovery of the principal component is generally desired.
We address this shortcoming via a novel proof for recovery that is applicable to sparse tensors and to the quantum setting. The quantum algorithm for recovery is especially natural for the qubit-based encoding. We note that recovery introduces an overhead of $n$, however, we expect this to be mitigated significantly with further work. The details of this extension are given in Section~\ref{sec:recovery}.

\paragraph{Algorithmic Extension: Quantum algorithms for asymmetric tensors}
All existing algorithms suffer from the key limitation of requiring symmetric input tensors. These are rare in practice, as tensors for real-world data typically represent different information in each dimension. Indeed, most classical algorithms for low-rank tensor PCA and completion naturally work on asymmetric tensors. It is natural to require quantum algorithms for tensor PCA to apply to this setting as well.
In Section~\ref{sec:asymm}, by means of a suitable ``symmetrization" argument, we show that an asymmetric tensor can be treated with the Kikuchi method at the cost of only constant factors, achieving detection. The proof is technically notable because the symmetrized tensor is sparse and structured, and to the best of our knowledge the Kikuchi method has not been demonstrated to work for structured tensors, beyond the work of Guruswami, Kothari and Manohar which considers smoothed constraint satisfaction problems~\cite{guruswami2022algorithms}. We also design a guiding state in this case that retains the large overlap on the high-energy subspace. In combination, these ingredients ensure an asymptotic quartic quantum advantage for asymmetric sparse tensor PCA. 

\subsection{Improved Classical Algorithms: Gupta \textit{et al.}~\cite{gupta2025classical}}
\label{sec:new-classical-algorithms}
A recent preprint by Gupta \textit{et al.}~\cite{gupta2025classical} describes a classical algorithm that achieves a quadratic speedup over the algorithm of Wein \textit{et al.}~\cite{wein2019kikuchi} when $k$ is sufficiently large. In this regime, the quantum algorithms described in this paper yield only a quadratic speedup over the best classical algorithm. The algorithm of~\cite{gupta2025classical} distinguishes the planted case from the null case by efficiently finding many even covers---a combinatorial structure---in a hypergraph derived from the Kikuchi matrix (termed the Kikuchi graph). This approach is fundamentally different from spectral methods and leverages combinatorial properties to improve runtime.

To ensure that the combinatorial algorithm works, the required level $\ell'$ of the Kikuchi hierarchy is generally different from the level $\ell$ required in~\cite{wein2019kikuchi}. The combinatorial algorithm has a runtime of $n^{(0.5\ell'+k)}$. For the quantum algorithm to maintain a quadratic speedup over this classical method, $\ell'$ must be comparable to $\ell$, which typically requires $k$ to be large. As indicated in~\cite{gupta2025classical}, the algorithm is not known to obtain any speedup for small $k$, such as $k=4$, in which case the quartic quantum speedup persists.

A notable strength of Gupta \textit{et al.}~\cite{gupta2025classical} is its applicability to the semirandom model, where the hypergraph structure can be adversarially chosen. This is a more general and challenging setting than the fully random case addressed by some quantum algorithms. However, a technical difference between the algorithms of~\cite{gupta2025classical} and~\cite{wein2019kikuchi} is that it is unclear how to perform recovery (i.e., finding the planted solution) using the combinatorial algorithm; the focus is on detection. Furthermore, \cite{wein2019kikuchi} only demonstrates recovery for the dense tensor PCA case. In contrast, we show in this paper (Sec.~\ref{sec:recovery}) that for sparse kXOR, in the spectral setting, recovery can be performed while retaining the quantum speedup.

\section{Quantum algorithm}\label{sec:quantum_algo}
\subsection{Notation and useful approximations}
In contrast with preceding works~\cite{hastings2020classical,schmidhuber2024quartic}, we will utilize the natural binary encoding mapping $n$ hard-core bosons to $n$ spins. A state of $k$ bosons will be supported on computational basis states of Hamming weight $k$. For instance, a multi-index $S$ will be mapped to the computational basis state $|S\>$ which is $1$ on the qubits in $S$ and zero everywhere else. 
We will denote multi-indices which are unordered and distinct by upper case letters (hence sets, e.g. $S \in {[n] \choose k})$, while we will use lower case if they are ordered and not necessarily distinct (hence tuples, e.g. $s \in [n]^k$).

The Kikuchi Hamiltonian of order $\ell$ for a symmetric tensor $T$ is
$$
    \KC_\ell = \sum_{U,V \in {[n] \choose \ell}} T_{U\Delta V}\, 1(U\Delta V \in \SC)\,|U\>\<V| = \sum_{S \in \SC} T_S \sum_{U,V \in {[n] \choose \ell}}1(U\Delta V =S)\,|U\>\<V|.
$$
As before we indicate with $\SC$ the set of multi-indices of observed entries of $T$, each multi-index being a set and hence unordered and with distinct indices.
Now we present some useful quantities alongside their approximation. 
The Kikuchi Hessian is the adjacency matrix of the weighted Kikuchi graph, whose vertex set is ${[n]\choose \ell}$.
Combinatorially one finds that the graph has $E$ edges, where $E = \frac{1}{2}m{n-k\choose \ell-k/2}{k\choose k/2}$. Alternatively we define $d_{n,k,\ell,m} = m\delta_{n,k,\ell}$ to be the average row sparsity of the Kikuchi Hessian, with
$$\delta_{n,k,\ell} = \frac{{n-k\choose \ell-k/2}{k\choose k/2}}{{n\choose \ell}} \le (1 -o(1)) {k\choose k/2}\left(\frac{\ell}{n}\right)^{k/2}.$$

For a dense tensor, the corresponding Kikuchi graph is $\Delta_{n,k,\ell}$-regular, where
$$\Delta_{n,k,\ell} = {n -\ell \choose k/2}{\ell \choose k/2}.$$
Therefore a crude bound for the spectral norm of the Kikuchi Hamiltonian is $\|\KC_\ell\| \le \Delta_{n,k,\ell}$. However for random or sparse planted tensors this is expected to be much less. In particular, we find that with high probability $\|\KC_\ell\| \le \mathrm{poly}(\ell,\log n)$ for both random and planted tensors when choosing $m = \tilde{\Omega}(n^{k/2})$~\cite{schmidhuber2024quartic}. 

We extensively use the asymptotic notation $O,o,\Omega,\omega,\Theta$. We often shorthand $A\sim B$ to mean $A = \Theta(B)$ (except when indicating the law of a random variable), and $A \approx B$ to mean equality up to leading order.

\subsection{Tensor model}
We use the tensor model that is adapted to planted kXOR~\cite{schmidhuber2024quartic}. This is primarily because the arguments and circuits are simpler to explain in this setting. It also allows us to directly leverage the guiding state preparation and detection from~\cite{schmidhuber2024quartic} in the context of the new embedding presented in this paper. We note however, that the proof of \emph{recovery} for the quantum algorithm is novel, and does not appear in prior work.

Within the spiked tensor model, we focus on signed symmetric tensors $T\in (\{\pm1 \}^n)^{\otimes k}$ of the form:
$$T = \vec{z}^{\otimes k}\odot \Xi$$
where we always assume that $k$ is even. $\Xi$ is a symmetric tensor of Skellam-distributed random variables, which are independent whenever the multi-indices are distinct up to permutation.
For any permutation $\pi$ and $i_1 <...<i_k \in [n]$, $\Xi_{\pi(i_1,...,i_k)} \sim \mathrm{Skellam}(\frac{1+\rho}{2}q, \frac{1-\rho}{2}q)$ with $q = \frac{m}{{n \choose k}}$ (we ignore the tensor multi-indices with repetitions as these do not contribute to the Kikuchi matrix). Recall that $\mathrm{Skellam}(\mu_1, \mu_2)$ is the distribution of the difference of two independent Poisson-distributed random variables $X-Y$, $X\sim \mathrm{Poi}(\mu_1)$, $Y\sim \mathrm{Poi}(\mu_2)$, and it has mean $\mu_1 - \mu_2$ and variance $\mu_1 + \mu_2$. This is equivalent to the following model for the tensor construction: pick the entries in $\SC$ by including each element of ${[n] \choose k}$ a number of times $\sim \mathrm{Poi}(q)$, such that we have an average of $m$ elements being included. Then for each element $S$ let the tensor entry be $T_S = z_S \sigma_S$, where $\sigma_S = (-1)^{B_S}$ is a random sign with $B_S \sim \mathrm{Bernoulli}(\frac{1-\rho}{2})$. Elements that may be sampled repeatedly are treated as independent entries when applying the random sign.

Using this model will be convenient for the proofs as done in~\cite{schmidhuber2024quartic}. However, for the construction of the quantum circuit we assume for simplicity that no element is included more than once, such that the tensor has entries in $\{\pm 1\}$. In practice having a collision is a low-probability event and at our regime it only happens on a constant number of entries, so we expect the proofs to still be valid.

\subsection{Overview of algorithm}\label{sec:overview}
The algorithm for quantum tensor PCA and completion falls in the family of guided ground state energy estimation and preparation, in the sense that a guiding state is provided which has a guaranteed overlap with the ground state. Ground state preparation is a key application of quantum computing for quantum chemistry and material science~\cite{mcardle2020quantum, cao2019quantum}, and as mentioned the guided version in some sense encompasses the full power of quantum computing depending on the overlap with the initial state~\cite{gharibian2022dequantizing}. Due to the relevance of ground state preparation, we are gifted with an extensive literature for algorithms for ground-state preparation, see e.g.~\cite{ge2019faster, lin2020near, dong2022ground, cubitt2023dissipative, motlagh2024ground}. 

For our specific implementation we will take inspiration from the asymptotically optimal procedures of Lin and Tong~\cite{lin2020near}. These rely on performing amplitude amplification using a reflector on the desired Hamiltonian eigenspace, which may in turn be prepared via quantum signal processing (QSP)~\cite{martyn2021grand, lin2020near}. Specifically, let $\Pi_{\ge \lambda^*}$ be the projector to the eigenspace of $\KC_\ell$ with eigenvalues larger than some $\lambda^*$ and $|\Gamma\>$ the guiding state. Define $\zeta^2 = \<\Gamma|\Pi_{\ge \lambda}|\Gamma\>$. Then with access to a guiding state preparation unitary $U_\Gamma$ and a block-encoding for the Kikuchi matrix $U_\KC$, we intend to prepare a unitary that performs amplitude amplification on the subspace defined by $\Pi_{\ge \lambda^*}$:
$$\left(\prod_{i=1}^L\Pi_{\Gamma} \Pi_{\ge \lambda^*}\right)U_\Gamma|0\>$$
with $L \sim \zeta^{-1}$.
However, regular amplitude amplification requires knowing the precise value of $\zeta$, for which we only have a lower bound $\zeta \ge \gamma_{n,k,\ell,m}$~\cite{schmidhuber2024quartic}.
Instead we can use the \textit{fixed-point amplitude amplification} (FPAA), which can be realized with the quantum singular value transformation (QSVT)~\cite{martyn2021grand}. The resulting algorithm is composed of two projector-controlled phase shifts parameterized by $\phi,\,\varphi$, repeated $L = O(\gamma_{n,k,\ell,m}^{-1})$ times:
$$\Pi_{\ge \lambda^*}^{\phi_0}\left(\prod_{i=1}^L\Pi_{\Gamma}^{\varphi_i} \Pi_{\ge \lambda^*}^{\phi_i}\right)U_\Gamma|0\>$$
where $\Pi^\phi = e^{i\phi(2\Pi - \mathbb{1})}$, using a QSVT sequence $\phi_0,(\phi_i,\varphi_i)_{i=1}^L$ that realizes an odd polynomial approximating the step function on the nonnegative reals, such that the singular value corresponding to the intended transformation is boosted close to $1$.

\subsection{State preparation}
The guiding state that we seek to prepare is
$$
    |\Gamma_\ell\> = \frac{1}{\chi} \sum_{(S_1,...,S_c) \in \SC^c \text{ disjoint}} T_{S_1} ... T_{S_c} |S_1 \oplus ... \oplus S_c\>
$$
where $\chi^2 = \sum_{(S_1,...,S_c) \in \SC^c} 1(S_1,...,S_c\, \text{disjoint})$ for a signed $T$.\\

\begin{proposition}
    We can prepare the state 
    $$
        |\phi\> \propto \sum_{S \in \SC} T_S|S\>
    $$
    using $O(mk \log m)$ gates.
\end{proposition}
\begin{proof}
    This can be done by the following method. First, prepare the dense state on $s=\lceil \log m\rceil$ qubits
    $$
    |s\> \propto \sum_{i,S\in \text{enum}(\SC)} T_S|i\>
    $$
which can be done using $O(2^s) = O(m)$ 
gates~\cite{malvetti2021quantum}.

Index the elements of $\SC$ with $i\in [m]$. Then one applies the following state preparation unitary:
$$U_{\SC} = \prod_{i\in[m]} U_{S_i \rightarrow i}\prod_{i\in[m]} U_{i \rightarrow S_i}.$$
$U_{i \rightarrow S_i}$ adds $S_i$ to the $n$-qubit register conditional to $i$ in the $s$-qubit register:
$$U_{i \rightarrow S_i}|i\>^s|x\>^n = |i\>^s|x + S_i\>^n.$$
This is composed of a $\mathrm{C}^s\mathrm{X}$ repeated $k$ times, where the control is active on $|i\>$ and writes the ON bits of $S_i$.
$U_{S_i \rightarrow i}$ adds $i$ to the $s$-qubit register conditional to the qubits of $S_i$ being ON in the $n$-qubit register:
$$U_{S_i \rightarrow i}|x\>^s|y\>^n = \begin{cases}
    |x + i\>^s|y\>^n &\text{ if } [y]_{S_i} = 1,\\
    |x\>^s|y\>^n &\text{ otherwise.}
\end{cases}$$
This is composed of a $\mathrm{C}^k\mathrm{X}$ repeated $s$ times, where the control is active on $|S_i\>$ and writes the ON bits of $i$.
Then we see that applying this unitary to $|s\>^s|0\>^n$ leads to $|0\>^s|\phi\>^n$. Since one $\mathrm{C}^r\mathrm{X}$ can be decomposed into $r-1$ Toffoli gates, both unitaries use $O(sk)$ primitive gates.
\end{proof}
In the signed case this can be further improved:
\begin{proposition}[Signed tensor entries]
    When $T_S = (-1)^{b_S} \in \{\pm 1\}$, we can prepare the state $|\phi\>$ using $<2msk+s$ gates. If $m \neq 2^{\lceil\log s\rceil}$ this circuit is repeated a constant number of times.
\end{proposition}
\begin{proof}
    Let $|s\> = \text{Had}^{s}|0\>$ and define $U_{i \rightarrow S_i}$ such that it performs the correct phase flip when $T_S=-1$, which can be done at the cost of a $\mathrm{Z}$ gate: 
    $$U_{i \rightarrow S_i}|i\>^s|x\>^n = (-1)^{b_i}|i\>^s|x + S_i\>^n.$$
    This procedure does incur a cost for when $m \neq 2^s$ since there are basis states in $|s\>$ that do not get reset to $0$. This cost is $\frac{2^s - m}{2^s} < \frac{1}{2}$ and the correct states can be indentified by the condition that the $s$-register is $0$, and so using a constant number of rounds of amplitude amplification suffices.
\end{proof}
The state preparation is repeated in parallel $c$ times, for a total of $O(cm k\log m) = O(m\ell\log m)$ gates, producing the state $|\phi\>^{\otimes c}$.\\
\begin{proposition}
    Let $\Pi_{\ell}$ be the projector to the Hamming weight-$\ell$ subspace. Define
    \begin{equation}
        \alpha_\ell := \|\Pi_\ell|\phi\>^{\otimes c}\|^2.
    \end{equation}
    Then for a signed tensor, if $m = \Omega(n^{k/2})$, with probability $\ge 1-1/n$ we have
    $$\alpha_\ell \ge1 - {c \choose 2} \left(\frac{k^2}{n} + \frac{4k\log n}{m}\right) = 1 - O\left(\frac{\ell^2}{n}\right).$$
\end{proposition}
\begin{proof}
        We follow the proof of Lemma 4.16 of \cite{schmidhuber2024quartic}. Seeing $\SC$ as a random $k$-uniform $m$-hyperedge hypergraph on $n$ vertices, the degree $d_v$ of a vertex $v$ can be treated with a multiplicative Chernoff bound
        $$P\Big(d_v \ge (1+\delta)\frac{km}{n}\Big) \le e^{-\frac{\delta^2}{\delta+2}\frac{km}{n}}.$$
        Hence one finds that
        $$P\Big(\max_{v}d_v \ge \frac{km}{n} + 4\log n\Big) \le \frac{1}{n}.$$
        Therefore with probability $\ge 1-\frac{1}{n}$ each hyperedge touches at most $\frac{k^2m}{n} + 4k\log n$ other hyperedges, and the probability that two hyperedges chosen at random touch is $\le \frac{k^2}{n} + \frac{4k\log n}{m}$. Thus among $c$ hyperedges chosen at random the probability that two touch is $\le {c \choose 2} (\frac{k^2}{n} + \frac{4k\log n}{m}) = O(\frac{\ell^2}{n})$, where we used $m = \Omega(n\log n)$.

        Now consider $|\phi\>^{\otimes c}$:
        \begin{equation}
            |\phi\>^{\otimes c} \propto \sum_{(S_1,...,S_c) \in {\SC^c}} T_{S_1} ... T_{S_c} |S_1 \oplus ... \oplus S_c\>
        \end{equation}
        The bitstrings $|S_1 \oplus ... \oplus S_c\>$ with Hamming weight $ \ell$ correspond exactly to the $c$-tuples from $\SC$ with no overlapping elements. Sampling one $c$-tuple from $\SC^c$ is clearly equivalent to sampling $c$ elements of $\SC$ at random, therefore the fraction of tuples with no overlapping elements is $\ge 1-O(\frac{\ell^2}{n})$ with probability $\ge1 - 1/n$.
        The proof is complete for a signed tensor all the entries are unit norm.
        \end{proof}
Therefore we see that enforcing $\ell = o(\sqrt{n})$ implies $\alpha_\ell = 1-o(1)$. As a result we get that $\chi^2 = \alpha_\ell m^c$.

\begin{proposition}
    Given the state $|\phi\>^{\otimes c}$, there exists a unitary $U_\Gamma$ acting on an $n$-qubit register and a $O(cn)$-qubit ancillary register $A$ that does
    \begin{equation}
        (\<0_A|\otimes \mathbb{1})U_\Gamma(|0_A\>\otimes|0\>)=\beta_\ell|\Gamma_\ell\>
    \end{equation}
    with $|\beta_\ell|^2 = \alpha_\ell c^{-\ell}$. This unitary has $O(\ell m \log m)$ gates.
\end{proposition}

\begin{proof}
    For the sake of discussion, consider now only the part of the state parallel to $\propto\Pi_\ell |\phi\>^{\otimes c}$, containing the bitstrings $|S_1\>...|S_c\>$ such that $S_1 \oplus ... \oplus S_c$ has total Hamming weight $\ell = ck$.
    Now we seek to combine the bitstrings $|S_1\>...|S_c\>$ that are disjoint into a single $|S_1 \oplus ... \oplus S_c\>$ of Hamming weight $\ell$, erasing the rest of the information about the ordering of the $c$ states. We do this with a probabilistic ``one-hot shuffling'' procedure. 
    
    First, we reorder the qubits in $n$ groups of $c$, the $i$th group having the $i$th qubit from all the registers. Notice that by enforcing disjointness between the $S_i$'s, we have ensured that each one of these groups has at most one qubit in the state $1$. 
    At this point we apply $D^{c \dagger}_1$ in parallel to each group, where $D^{c}_1$ is the unitary that prepares the Dicke state of Hamming weight 1 on $c$ qubits when applied to the state $|100...0\>$. As shown in \cite{bartschi2019deterministic}, this unitary can be built in the following manner:
    \begin{equation}
        D^{c}_1 = \prod_{i=1}^{c-1} G_{i,i+1}\left(\frac{1}{\sqrt{c+1-i}}\right)
    \end{equation}
    where $G_{i,j}(a)$ is a Givens rotation applied to qubits $i,j$ of angle $\arccos a$, that does:
    \begin{equation}
        G_{i,j}(a)|1_i0_j\>=a |1_i0_j\> + \sqrt{1-a^2} |0_i1_j\>
    \end{equation}
    and leaves the $|00\>$, $|11\>$ states unchanged. Therefore applying $D^{c \dagger}_1$ yields the state $|100...0\>$ with probability $1/c$ for any computational basis input of Hamming weight 1. Conversely, it is easy to see that if the starting state is $|0...0\>$ the output state is also all zeros with probability 1. Therefore postselecting on the qubits 2 to $c$ being 0 yields 1 on qubit 1 with probability $1/c$ if the input state is Hamming weight 1, and 0 if it is Hamming weight 0 with probability 1, which is the intended summation operation. So we apply this procedure in parallel $n$ times, succeeding with probability $1/c^\ell$. The total gate count is $O(cn)$.

    Note that in our construction $D^{c \dagger}_1$ is Hamming weight preserving, therefore enforcing the Hamming weight $\ell$ condition can be deferred to after this circuit. If the Hamming weight of the input to $D^{c \dagger}_1$ is $> 1$ then at least one of the measured qubits will hold a value of 1, and so the postselection will fail on these states. Combining the success probability of the one-hot shuffling with the probability of observing a Hamming weight $\ell$ input state we obtain a total success probability of $\alpha_\ell c^{-\ell}$.
    \end{proof}
    
Recall that the guiding state is
\begin{equation}
    |\Gamma_\ell\> = \frac{1}{\sqrt{\alpha_\ell m^c}} \sum_{(S_1,...,S_c) \text{ disjoint}} T_{S_1} ... T_{S_c} |S_1 \oplus ... \oplus S_c\>.
\end{equation}
We have thus provided a state preparation unitary that does
$$U_\Gamma|0\>^{nc}=\sqrt{\alpha_\ell c^{-\ell}}|0\>^{nc-n}|\Gamma_\ell\> + |\text{trash}\>$$
where $\<\text{trash}|(|0\>\<0|^{nc-n} \otimes \mathbb{1}_{2}^n)|\text{trash}\>=0$ and $\sqrt{\alpha_\ell c^{-\ell}}$ can be calculated from the tensor.
Therefore, using amplitude amplification, we can prepare $|\Gamma_\ell\>$ deterministically using $r \sim \sqrt{c^{\ell}/\alpha_\ell}$ calls to the state preparation unitary. This brings the total gate count to $O(\ell c^{\ell/2} m\log n)$, which is a significant boost from $O(\ell^{-\ell/2})$ of previous methods. 

In the overall circuit the unitary $U_\Gamma$ is also used to prepare the projector-controlled phase shifts $\Pi_\Gamma^{\varphi_i}$, by sandwiching a controlled $z$-rotation between $U_\Gamma^\dagger$ and $U_\Gamma$. Since $U_\Gamma$ block-encodes the correct state-preparation unitary scaled by a constant, we can utilize oblivious amplitude amplification~\cite{berry2014exponential} to perfectly prepare the phase shift, again with a gate count of $O(\ell c^{\ell/2} m\log n)$.

\subsection{Phase estimation}\label{sec:PE}

We assume that we have access to a circuit $\mathrm{Proj}$ that block-encodes $\Pi_{\ge \lambda^*}$ with an ancillary system $a$:
$$(\<0|^{a}\otimes I)\mathrm{Proj}(|0\>^{a}\otimes I) =  \Pi_{\ge \lambda^*}.$$
Then by using a multi-controlled $z$-rotation on the ancilla $\mathrm{C}^{a}R_\phi = e^{i\phi(2|0\>\<0|^a-\mathbb{1}^a)}$ we can implement $\Pi_{\ge \lambda^*}^\phi$:
$$(\<0|^{a}\otimes \mathbb{1})\mathrm{Proj}^\dagger\,\mathrm{C}^{a}R_{2\phi}\, \mathrm{Proj}(|0\>^{a}\otimes \mathbb{1}) = (e^{i2\phi}-1)\Pi_{\ge \lambda^*} +\mathbb{1} = e^{i\phi}\Pi_{\ge \lambda^*}^\phi.
$$
As outlined in~\cite{lin2020near}, $\mathrm{Proj}$ can be implemented from a single controlled reflector and two Hadamard gates. The reflector in turn can be constructed with the quantum eigenvalue transform starting from a block encoding of the Hamiltonian~\cite{martyn2021grand}.
More precisely, to implement the reflector for the high-energy eigenspace of $\KC$ to error $\varepsilon$ and with eigenvalue gap $\delta$, having access to a block encoding $U_{\KC}$ of $\KC/\alpha$ such that $\|\KC\| \le \alpha$, the reflector requires a QSP sequence of length
$$
    q = O\left(\frac{\alpha}{\delta}\log(\varepsilon^{-1})\right).
$$
The relevant QSP polynomial is symmetric around zero, which requires that no eigenvalue exists below $-\lambda^*$. To ensure this, we use a controlled version of $U_{\KC}$ to block-encode the positive semidefinite $\frac{1}{2}(\KC/\alpha + \mathbb{1})$ and place the threshold at $\frac{1}{2}(\lambda^*/\alpha + 1)$~\cite{martyn2021grand}.
Therefore we can easily get a block encoding of $\Pi_{\ge \lambda^*}^\phi$ by using $2q$ repetitions of $\mathrm{C}U_{\KC}$ and $\mathrm{C}U_{\KC}^\dagger$.

To implement $U_{\KC}$ we use a quantum-walk-style block encoding with two oracles: an evaluation oracle $O_E$ and an adjacency oracle $O_A$~\cite{berry2009black}. The circuit takes the form shown in Fig.~\ref{fig:oracle-overview}.
\begin{figure}
    \centering
    \includegraphics[width=0.9\linewidth]{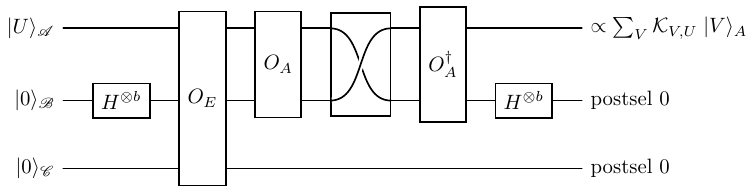}
    \caption{Overview of the block-encoding construction, inspired by quantum walks~\cite{berry2009black}.}
    \label{fig:oracle-overview}
\end{figure}
The evaluation oracle acts on two work registers $\mathscr{A}$ and $\mathscr{B}$ and a single-qubit ancillary register $\mathscr{C}$ initialized at 1. If the $\mathscr{A}$ register encodes the tensor index $U$ and the $\mathscr{B}$ register holds value $k$ and $k\le \sigma(U)$, the sparsity of column $U$ in $\KC$, then the value of the $k$th nonzero entry of the column is amplitude-encoded in the ancilla, otherwise the ancilla is left at 1.
\begin{align}
    O_E: |0, k, U\> \mapsto \begin{cases}
        \KC_{V(k, U), U}|0, k, U\> + |1\>|\text{trash}\> &\text{ if } k \le \sigma(U); \\
        |1\>|\text{trash}\> &\text{ otherwise.}
    \end{cases} 
\end{align}
Here, $V(k,U) = U \Delta S(k,U)$ where $S(k,U)$ is the $k$th $S \in \SC$ such that $|U \Delta S|=\ell$. Note that by slight abuse of notation we introduced a fixed ordering to $\SC$.
The adjacency oracle takes as input registers $\mathscr{A}$ and $\mathscr{B}$ and is defined as 
\begin{equation}
    O_A: |k, U\> \mapsto \begin{cases}
        |V(k, U), U\> &\text{ if } k \le \sigma(U); \\
        |\text{trash}\> &\text{ otherwise.}
        \end{cases}
\end{equation}
Notice that as defined this is an isometry, since register $\mathscr{B}$ is mapped to a larger register $\mathscr{B}'$ of the same size as $\mathscr{A}$ ($n$ qubits). Register $\mathscr{B}$ only needs to encode values from 1 to $\max_U \sigma(U) = d_{\max}$, the maximum degree of the Kikuchi graph of $T$. We then choose the register to have $b = \lceil \log_2 d_{\max}\rceil$ qubits, $d_{\max}$ can be easily calculated directly from the tensor without constructing the Kikuchi matrix.
One can check that 
$$\<0, 0, V|U_{\KC}|0, 0, U\> = \frac{\KC_{V, U}}{2^b}$$
and so we can do block encoding with the above circuit. Note that $\|\KC\| \le d_{\max} \le 2^b$ so the resulting block-encoded matrix has spectral radius $\le 1$.

Both oracles $O_A$ and $O_E$ have similar constructions, and in fact they can largely be combined. We present the full circuit construction in Appendix~\ref{app:oracles}, and we briefly summarize it here.

First, the $\mathscr{B}$ register is initialized in the uniform superposition over all values of $k$ from $1$ to $d_{\max}$, and the value of register $\mathscr{A}$ is copied in a register $\mathscr{E}$. 
The circuit then is composed of four blocks, each with depth $\sim m$. 

The first block iterates through all $S \in \SC$. Each time $|S \Delta U| = \ell$ an ancillary counter register $\mathscr{D}$ is incremented by 1, after which the value of the counter is compared with the value $k$ in register $\mathscr{B}$. If they match, the bits of the current $S$ are flipped in register $\mathscr{E}$. If $\mathscr{B},\mathscr{E}$ have initial value $k,U$ then $S=S(k,U)$ and $\mathscr{E}$ goes to $S(k,U) \Delta U = V(k,U)$. At the same time, register $\mathscr{C}$, initialized at $|1\>$, is $y$-rotated encoding the value of  $\KC_{V(k,U), U} = T_{S(k,U)}$. 

The second block iterates backwards through $S \in \SC$, each iteration comparing the value of $A\Delta S$ and $\mathscr{E}$. If they agree, the value of the counter register $\mathscr{D}$ is copied over in register $\mathscr{B}$, thus erasing the value of $k$ there. In each iteration the same check as before $|S \Delta  U| \overset{?}{=} \ell$ is also performed, this time decrementing the value of register $\mathscr{D}$ all the way to 0.
Then the swap is applied between $\mathscr{A}$ and $\mathscr{E}$, and two blocks follow that are almost the inverse of the first two blocks, with the exception that the last block does not act on the register $\mathscr{C}$. Then the register $\mathscr{E}$ is uncomputed by XORing the value of register $\mathscr{A}$, and a final set of Hadamards is applied to register $\mathscr{B}$.
Overall, registers $\mathscr{D}$ and $\mathscr{E}$ are fully uncomputed, and the block encoding to $\mathscr{A}$ succeeds if $\mathscr{B}$ and $\mathscr{C}$ are measured to be 0.

As mentioned before, we will need controlled versions of these block encoding circuits. These are achieved simply by controlling the rotation on the ancilla $\mathscr{C}$ and the swap, since the rest of the circuit factors out. Therefore the control only adds $O(n)$ non-Clifford gates and depth $O(\log n)$ which are negligible compared to the overall circuit.

\section{Analysis of detection}
\label{sec:resource_est}
\subsection{Theory}
Here we summarize the theoretical findings from Schmidhuber \textit{et al.}~\cite{schmidhuber2024quartic} that we will need in order to establish the thresholds for detection and the number of amplitude amplification repetitions.

The following results upper bounds the degree of the Kikuchi matrix for a tensor with randomly masked entries.
\begin{proposition}[Prop.2.15 in \cite{schmidhuber2024quartic}]\label{prop:degree}
    The maximum degree of the Kikuchi matrix $\KC_\ell$ is
    $$d_{\max} \le (1+\kappa)d_{n,k,\ell,m}$$
    for any $\kappa >0$, except with probability at most ${n \choose \ell}e^{-\frac{\kappa^2}{2+\kappa}d_{n,k,\ell,m}}$.
\end{proposition}
With $m \sim n^{k/2} \log n$, assuming $d_{n,k,\ell,m} \gg \log {n \choose \ell}$ we can choose $\kappa$ to be constant and get that with high probability $d_{\max} = O(\ell^{k/2} \log n)$.\\

\begin{proposition}[Prop.2.16 in \cite{schmidhuber2024quartic}]\label{prop:2.16}
    The spectral norm of a Kikuchi matrix with Rademacher entries is 
    $$\|\mathcal{K}_{\text{rand}}\| \le \sqrt{2(1+\kappa)(1+\epsilon)d_{n,k,\ell,m}\log {n \choose \ell}}$$
    for any $\kappa > 0$ and $\epsilon > 0$, except with probability at most ${n \choose \ell}e^{-\frac{\kappa^2}{2+\kappa}d_{n,k,\ell,m}} + {n \choose \ell}^{-\epsilon}$.
\end{proposition}
With the parameter choices from before, and choosing $\epsilon$ constant, with high probability for a Rademacher Kikuchi matrix $\|\KC_{\text{rand}}\| = O(\ell^{(k+2)/4} \log n)$.\\

\begin{proposition}[Prop.2.30 in \cite{schmidhuber2024quartic}]\label{prop:2.30}
    For the planted case there exists an eigenvalue
    $$\lambda \ge (1-\gamma)\rho d_{n,k,\ell,m}$$
    for any $0 < \gamma <1$, except with probability at most $e^{-\gamma^2\rho^2m/2}$.
\end{proposition}
Therefore this shows the existance of eigenvalues of energy $\Omega(\ell^{k/2}\log n)$, a considerable separation from the random case for large $\ell$. This sets the threshold energy for detection.

Finally, we use the following result from~\cite{schmidhuber2024quartic} to lower bound the guiding state overlap. Note that this result is in a slightly stronger form than what was shown in the original paper.
\begin{proposition}[Thm.2.40 in \cite{schmidhuber2024quartic}, tightened]\label{prop:guiding}
    In the planted case, the guiding state support on the subspace of energy $\ge \lambda^*= (1-\gamma)\rho d_{n,k,\ell,m}$ is lower bounded by
    \begin{equation}
            \<\Gamma|\Pi_{\ge\lambda^*}|\Gamma\> \ge \xi \left(\frac{m}{{n \choose k}}\right)^{\ell/k}, \;\;\; \xi=\frac{\ell!}{(\ell/k)!(k!)^{\ell/k}} \frac{\rho\epsilon\nu}{4A}(\rho^2\zeta)^{\ell/k}
    \end{equation}
    except with probability at most
        $${n \choose \ell}e^{-\frac{\kappa^2}{2+\kappa}(1-\zeta)d} +e^{-\frac{(\gamma-\epsilon)^2\rho^2(1-\zeta)m}{2}}+\frac{{n \choose \ell -k}{n\choose k}}{{n\choose \ell}{\ell \choose k}}\frac{8.16(\ell/k)^2A}{\zeta\epsilon\rho^3m}+\nu,$$
    where $A = 1+\kappa - (1-\gamma)\rho$, $0<\epsilon < \gamma$, $0< \nu,\zeta<1$ are constants.
\end{proposition}
To prove this, we first tighten the following result:
\begin{proposition}[Prop.2.31~\cite{schmidhuber2024quartic}, tightened]\label{prop:2.31tight}
    Let $\kappa, \gamma,\hat\gamma$ be positive constants, $0 < \hat\gamma < \gamma$, and define $\epsilon := \gamma - \hat{\gamma}$. Let $\Pi_{\ge \lambda^*}$ be the projector to a subspace of eigenvalue at least $\lambda^* = (1-\gamma)\rho d$. Then
    \begin{equation}
        \<z^{\odot\ell}|\Pi_{\ge \lambda^*}|z^{\odot\ell}\>  \ge \frac{\epsilon\rho}{1+\kappa - (1-\gamma)\rho}
    \end{equation}
    except with probability at most $e^{-\frac{\hat\gamma^2\rho^2}{2}m} + {n \choose \ell}e^{-\frac{\kappa^2}{2+\kappa}d}$.
\end{proposition}
\begin{proof}
    Define a random variable $X$ on the spectrum of $\mathcal{K}$ which takes value $\lambda$ with probability $|\<\lambda|z^{\odot \ell}\>|^2$, such that $\<X\> = \<z^{\odot\ell}|\mathcal{K}|z^{\odot\ell}\>$, which from Prop.~\ref{prop:2.30} is $ \ge (1-\hat\gamma)\rho d$ except with probability at most $e^{-\frac{\hat\gamma^2\rho^2}{2}m}$. Also notice that in fact $\<z^{\odot\ell}|\Pi_{\ge \lambda^*}|z^{\odot\ell}\> = \text{Pr}(X \ge \lambda^*)$. Now rescale it to $\hat{X} =\lambda_{\max}-X$ such that it is nonnegative, and apply Markov's inequality, to obtain for all $x > 0$,
    $$\text{Pr}(\hat{X} \ge x) \le \frac{\lambda_{\max}  - \<z^{\odot\ell}|\mathcal{K}|z^{\odot\ell}\>}{x} \implies \text{Pr}(X > \lambda_{\max}-x) \ge \frac{x -\lambda_{\max}  + \<z^{\odot\ell}|\mathcal{K}|z^{\odot\ell}\>}{x}.$$
    Now let $x = \lambda_{\max}-\lambda^*$, therefore we get that 
    $$\<z^{\odot\ell}|\Pi_{\ge \lambda^*}|z^{\odot\ell}\> \ge \frac{(1-\hat\gamma)\rho d - \lambda^*}{\lambda_{\max} - \lambda^*} \ge \frac{\epsilon\rho d }{\lambda_{max} - (1-\gamma)\rho d}$$
    except with probability at most $e^{-\frac{\hat\gamma^2\rho^2}{2}m}$.
    We now use $\lambda_{\max} \le \|\mathcal{K}\|\le d_{\max}$, and via Prop.~\ref{prop:degree}, for all $\kappa >0$, with probability $\ge 1- {n \choose \ell}\exp(-\frac{\kappa^2}{2+\kappa}d)$, $d_{\max} \le (1+\kappa) d$.
    Therefore we get the intended statement.
\end{proof}
Now we can prove Prop.~\ref{prop:guiding}:
\begin{proof}
    From Thm.2.36~\cite{schmidhuber2024quartic}, with $|v\> \propto \Pi|z^{\odot\ell}\>$ normalized and $|\Gamma\>$ as defined in \cite{schmidhuber2024quartic} (not normalized), we have 
    $$\<v|\Gamma\> \ge \frac{1}{2} \sqrt{\text{Part}_k(\ell)} (\rho\sqrt{q})^c \sqrt{\<z^{\odot\ell}|\Pi|z^{\odot \ell}\>}$$
    except with probability at most 
    $$\frac{8.16(\ell/k)^2{n\choose \ell-k}}{{\ell\choose k}{n \choose \ell}}\frac{1}{ (\rho\sqrt{q})^2 \<z^{\odot\ell}|\Pi|z^{\odot \ell}\>^2}.$$
    From Prop.~\ref{prop:2.31tight}, we have
    $$\<z^{\odot\ell}|\Pi|z^{\odot \ell}\> \ge \frac{\epsilon\rho}{1+\kappa - (1-\gamma)\rho}$$
    except with probability at most
    $$e^{-\frac{\hat\gamma^2\rho^2}{2}m} + {n \choose \ell}e^{-\frac{\kappa^2}{2+\kappa}d}.$$
    Also we get from \cite{schmidhuber2024quartic} that except with probability at most $\nu$, $\<\Gamma|\Gamma\> \le 1.0202\nu^{-1}$.
    Combining these and letting $m \rightarrow \zeta m$ for the guiding state and $m \rightarrow (1-\zeta) m$ for the Kikuchi matrix we get the result.
\end{proof}

With amplitude amplification we perform a number of repetitions
$$L\sim\frac{1}{\sqrt{\<\Gamma|\Pi_{\ge \lambda^*}|\Gamma\>}} \le \frac{1}{\sqrt{\xi \left(\frac{m}{{n \choose k}}\right)^{\ell/k}}},$$
which is $\tilde O(n^{\ell/4})$ when $m\sim n^{k/2}\log n$.

\subsection{Resource estimation}
We present the workflow for our resource estimation for a detection task. In the planted case, Prop.~\ref{prop:guiding} above ensures that the guiding state has a large support in the subspace of energy $\lambda^*(\gamma)$ for some tunable constant $\gamma$, with some $\gamma$-dependent probability. Supposing that $\rho$ is known, we can choose a $\gamma$ that gives a good probability. Then we set the phase estimation parameters $\lambda$ and $\delta$ such that $\lambda+\delta = \lambda^*(\gamma)$ and $\lambda-\delta = c\lambda^*(\gamma) = \omega(\ell^{(k+2)/4})$ for some constant $0<c < 1$, so that the procedure succeeding means that with high probability we are in the planted setting. Thus both parameters are $\sim \lambda^*$. Meanwhile, the block-encoded matrix is scaled by $\alpha =2^b \sim d_{\max}$, which by Prop.~\ref{prop:degree} is of similar order. Therefore $\lambda^*/\alpha \sim 1$ and $\delta/\alpha \sim 1$ meaning that $q$ only depends indirectly on the problem parameters via its dependence on $\varepsilon$ and the requirement that $L\varepsilon \ll 1$.

In the Appendix~\ref{app:circuit}, we present detailed constructions of the oracles. A comprehensive optimization of the quantum resources was performed with the objectives of decreasing non-Clifford gate count and depth, such as parallelizing operations when possible. Significant savings were obtained by utilizing the sparse and random structure of the tensor, allowing for simultaneous encoding of tensor entries in both the state preparation and phase estimation subroutines.

Table~\ref{tab:quantum_resources_main} gives a representative gate count obtained using realistic bounds for the guiding state overlap and the QSP sequence length, using $k=4$, $\ell=16$, $\rho = 1/4$ and $m = 10n^2\log n$.

\begin{table}[]
\begin{center}
\begin{tabular}{|c|c|c|c|c|c|c|c|c|c|}
\hline
$\boldsymbol{n}$ & \makecell{\textbf{Logical}\\ \textbf{Qubits}} & \makecell{Amp. amp. \\Repetitions} & \makecell{Depth \\PE \\
$\times 10^9$} & \makecell{Depth \\State \\
$\times 10^9$} & \makecell{Gates\\ PE \\
$\times 10^{12}$} & \makecell{Gates\\ State \\
$\times 10^{12}$} & \makecell{\textbf{Total}\\ \textbf{Depth} \\
$\times 10^{12}$ }& \makecell{\textbf{Total} \\\textbf{Gates} \\
$\times 10^{15}$} & \makecell{\textbf{Classical} \\\textbf{FLOPs} \\
$\times 10^{20}$} \\
\hline
60 & \totalqubitszero & \numrepetitionszero & \depthqspzero & \depthstatezero & \gateqspzero & \gatestatezero & \totaldepthzero & \totalgateszero & \classicalflopszero \\
\hline
80 & \totalqubitsone & \numrepetitionsone & \depthqspone & \depthstateone & \gateqspone & \gatestateone & \totaldepthone & \totalgatesone & \classicalflopsone \\
\hline
100 & \totalqubitstwo & \numrepetitionstwo & \depthqsptwo & \depthstatetwo & \gateqsptwo & \gatestatetwo & \totaldepthtwo & \totalgatestwo & \classicalflopstwo \\
\hline
120 & \totalqubitsfive & \numrepetitionsfive & \depthqspfive & \depthstatefive & \gateqspfive & \gatestatefive & \totaldepthfive & \totalgatesfive & \classicalflopsfive \\
\hline
\end{tabular}
\end{center}
 \caption{Quantum resource estimates for sparse spiked tensor PCA (detection) at various problem sizes, including logical qubits, total depth and gates. We also include the number of amplitude amplification repetitions, and the depth and gate count of the phase estimation (PE) and guiding state preparation circuits. We compare with the estimated FLOPs required by the classical method, which is the power method with the full Kikuchi matrix. For reference, note that the FLOPs used to train various large GPT models range from $10^{23}$ to $10^{25}$, according to folklore. Although, it is unlikely that such classical computational resources would ever be devoted to tensor PCA.}
 \label{tab:quantum_resources_main}
\end{table}

\section{Recovery for tensor completion}\label{sec:recovery}
In this section we detail on the recovery task along with a novel proof applicable to the setting of tensor completion. From a high-energy eigenvector $\vec{v} \in \mathbb{R}^{n\choose\ell}$ of $\KC_\ell$ with energy $\ge \lambda^*$, construct the voting matrix $V(\vec{v}) \in \mathbb{R}^{n\times n}$~\cite{wein2019kikuchi} as
\begin{equation}
    V_{ij}(\vec{v}) = \begin{cases}
        \sum_{U,V \in {[n]\choose \ell}} v_Uv_V 1(U\Delta V = \{i,j\}) &\text{ if } i \neq j,\\
        0 &\text{ otherwise.}
    \end{cases}
\end{equation}
If we have access to the quantum state $|v\>$ in qubit encoding, we see that this is
$$V_{ij} = \<v|g_{ij}|v\>,\;\;i\neq j,$$
where $g_{ij} = \frac{X_iX_j+Y_iY_j}{2}$ is the generator of a Givens rotation between qubits $i$ and $j$, and this can be done in one go with Bell basis measurements. Since the $n(n-1)/2$ measurements can be parallelized, this comes at a cost of $O(n)$ iterations. We do not exclude that in an actual execution of the quantum algorithm this can be brought down to a constant using randomized linear algebra results, since recovering the leading eigenvalue of $V$ is seen numerically (Sec.~\ref{sec:numerics}) to give a high overlap with the correct solutions (even though we will use a different approach in our proof of strong recovery).

We now prove weak and strong recovery, as defined in~\cite{wein2019kikuchi}. Define the correlation between a proposed solution $\vec{x}$ and the planted spike $\vec{z}$ as 
$$\text{corr}(\vec{x},\vec{z}) = \frac{|\vec{x}^\intercal \vec{z}|}{\|\vec{x}\|\|\vec{z}\|}.$$
Then weak recovery requires the existence of an algorithm outputting a solution with $\text{corr}(\vec{x},\vec{z}) = \Omega(1)$, while strong recovery needs $\text{corr}(\vec{x},\vec{z}) = 1-o(1)$, asymptotically in $n$.

\subsection{Proof of weak recovery}
The following theorem implies weak recovery.
\begin{theorem}
    The voting matrix $V$ formed from a state with energy $\ge \lambda^* := (1-\gamma)\rho m\delta$, for $0<\gamma<1$, satisfies with high probability
    $$\frac{\vec{z}^\intercal V \vec{z}}{n} \ge 1- \gamma-o_{n,\ell}(1),$$
    whenever $m = \Omega(n^{k/2}\log n)$ and $n > \ell/k^2$.
\end{theorem}
To prove the theorem, we first need the following proposition:
\begin{proposition}\label{prop:bernstein}
For all $t \le q{n-\ell\choose k/2}{\ell\choose k/2}$, we have
        $$\mathbb{P}(\|\rho q\KC^* -\KC\|\ge t) \le 2{n \choose \ell}e^{-t^2/4\sigma^2},\;\;\;\sigma^2 =q{n-\ell\choose k/2}{\ell\choose k/2}.$$
\end{proposition}
\begin{proof}
    Recall that we can write $T_S = z_S\Xi_S$ where $\Xi_S\sim \mathrm{Skellam}(\frac{1+\rho}{2}q, \frac{1-\rho}{2}q)$ with $q = \frac{m}{{n \choose k}}$, and that $\mathbb{E}\,\Xi_S = \rho q$ and $\text{Var}\,\Xi_S=q$.
    Note that
    $$
        \rho q\KC^* -\KC =  \sum_{S} \rho q (z_S - T_{S}) A_S = \sum_{S} (\rho q - \Xi_S)z_SA_S
    $$
    and therefore this is a sum of the form $\sum_S \xi_S M_S$, with $\xi_S = \rho q - \Xi_S$ zero mean iid random variables and $M_S = z_SA_S$ matrices obeying $M_S^2=\mathbb{1}$ and therefore $\|M_S\| = 1$.
    The Skellam distribution and therefore the distribution of $\xi_S$ is subexponential. More precisely, using standard bounds for the centered Poisson distribution, we find that the moment generating function for $X_S = \xi_S M_S$ obeys the Bernstein condition
    $$\log \mathbb{E}\,e^{\lambda X_S} \preceq \frac{q\lambda^2}{2(1-|\lambda|)} A_S^2 \;\;\;\forall\,|\lambda| \le 1.$$
    Therefore using a standard matrix Bernstein inequality (see~\cite{wainwright2019high}, Thm. 6.17), we have
    $$\mathbb{P}(\|\rho q\KC^* -\KC\|\ge t) \le 2{n \choose \ell}e^{-t^2/2(\sigma^2+t)},\;\;\;\sigma^2 = \|\sum_S\text{Var} \,\xi_S A_S \|.$$
    Now $\|\sum_S\text{Var} \,\xi_S A_S \| = q\|\sum_SA_S^2 \|$. Recall that $A_S$ is a Kikuchi matrix corresponding to a single tensor entry $S$ and it squares to a diagonal matrix of degrees of the corresponding Kikuchi graph. Since the sum is over all possible graph vertices, we get that $\sigma^2 = q\Delta_{n,k,\ell}$. The result follows.
\end{proof}
Therefore for all vectors $|v\>$, we have that with high probability over the randomness of $\KC$,
$$\rho q\<v|\KC^*|v\> \ge \<v|\KC|v\> - t\sqrt{q{n-\ell\choose k/2}{\ell\choose k/2}}$$
for $t$ say $\sqrt{40+\log 2{n\choose\ell}} < q{n-\ell\choose k/2}{\ell\choose k/2}$.
Crucially, this is true even for $|v\>$ that depends on $\KC$, like in our case. Now we can prove the theorem.
\begin{proof}
Let 
$$\KC = \sum_{S\in \mathcal{S}}T_S\sum_{U,V \in {[n]\choose \ell}} |U\>\<V|\, 1(U \Delta V = S),$$
$$\KC^* = \sum_{S\in {[n]\choose k}}z_S\sum_{U,V \in {[n]\choose \ell}} |U\>\<V|\, 1(U \Delta V = S).$$
Now we relate $\<v|\KC^*|v\>$ to $\vec{z}^\intercal V\vec{z}$.
\begin{proposition}
If $n > \ell/k^2$,
    $$\vec{z}^\intercal V\vec{z} \ge n\frac{\<v|\KC^*|v\>-  (-1)^{k/2}{\ell \choose k/2}}{{n-\ell\choose k/2}{\ell \choose k/2} - (-1)^{k/2}{\ell \choose k/2}} - \ell.$$
\end{proposition}
\begin{proof}
We can assume $\vec{z} = 1$. This is always possible by changing the basis of $\KC^*$ with a diagonal matrix $D$, $D_{U,U} = z_U$ thus obtaining the Kikuchi Hamiltonian with spike $1$. $D$ is unitary so we can identically rotate the vector $|v\>$, as well as $\KC$ while preserving expectation values. Then
$$\<v|\KC^*|v\> = \sum_{U,T} v_U v_T 1(|U\Delta T| = k)$$
and
$$\vec{z}^\intercal V\vec{z} = \sum_{U,T} v_U v_T 1(|U\Delta T| = 2).$$
These are both related to graphs in the $(n,\ell)$-Johnson association scheme~\cite{schrijver2003comparison,godsil2006association}. If $\mathcal{J}_i$ is the adjacency matrix of the distance-$i$ graph of this scheme, we see that $\<v|\KC^*|v\> = \<v|\mathcal{J}_{k/2}|v\>$ and $\vec{z}^\intercal V\vec{z} = \<v|\mathcal{J}_{1}|v\>$. By properties of the scheme, the matrices can be simultaneously diagonalized into a common set of eigenspaces labeled by $0 \le r \le \ell$ with eigenvalues given by the Eberlein polynomials:
$$\lambda_{r}(n,\ell,i) = \sum_{j=0}^{\min(r,i)}(-1)^j{r \choose j}{\ell-r \choose i-j}{n-\ell-r \choose i-j}.$$
If we let $p(r)$ be the mass of $|v\>$ on the $r$th eigenspace, we have
$$\<v|\KC^*|v\> = \sum_{r=0}^\ell \lambda_r(n,\ell,k/2)\,p(r),\;\;\;
\vec{z}^\intercal V\vec{z} = \sum_{r=0}^\ell \lambda_r(n,\ell,1)\,p(r).
$$
We now note the following facts, proven in~\cite{wein2019kikuchi}: if $n > \ell/k^2$:
\begin{itemize}
    \item If $r \le \ell - i$, the first term ($j=0$) is positive and nonzero. If $r > \ell - i$ the first nonzero term is $j= i-\ell+r$.
    \item Beyond the first nonzero term the magnitude of the subsequent terms (with alternating signs) is strictly decreasing. 
    \item Therefore, $\lambda_r$ has the same sign as the first nonzero term and its magnitude is bounded by that term's magnitude.
    \item For $r \le \ell-i$ successive eigenvalues are strictly decreasing in magnitude.
\end{itemize}
Therefore it follows the first negative eigenvalue is at $r = \ell-i+1$. For $i=1$, $\lambda_r \ge 0\;\forall r\le \ell-1$ and the only negative eigenvalue is $\lambda_{\ell}(n,\ell,1) = -\ell$. Also the positive eigenvalues with $0 \le r \le \ell-1$ are strictly decreasing in magnitude. Therefore
$$
\vec{z}^\intercal V\vec{z} \ge \lambda_{\ell-1} \sum_{r<\ell} p(r) + \lambda_{\ell} p(\ell) = (n-\ell)(1-p(\ell)) - \ell p(\ell) = n-\ell - np(\ell).
$$
We now upper bound $p(\ell)$ using the value of $\<v|\KC^*|v\>$. $p(\ell)$ is maximized when the remaining $1-p(\ell)$ of the mass is on the largest eigenvalue, at $r=0$. Therefore
$$p(\ell) \le \frac{\lambda_0(n,\ell,k/2) - \<v|\KC^*|v\>}{\lambda_0(n,\ell,k/2) - \lambda_\ell(n,\ell,k/2)} \implies \vec{z}^\intercal V\vec{z} \ge n\frac{\<v|\KC^*|v\>- \lambda_\ell(n,\ell,k/2)}{\lambda_0(n,\ell,k/2) - \lambda_\ell(n,\ell,k/2)} - \ell.$$
Plugging in the values of $\lambda_0(n,\ell,k/2),\, \lambda_\ell(n,\ell,k/2)$ gives the intended result.
\end{proof}
Combining with the previous result gives
$$\vec{z}^\intercal V\vec{z} \ge n\frac{\frac{1}{\rho q} \lambda - \frac{t}{\rho \sqrt{q}}\sqrt{{\ell \choose k/2}{n-\ell \choose k/2}} -  (-1)^{k/2}{\ell \choose k/2}}{{n-\ell\choose k/2}{\ell \choose k/2} - (-1)^{k/2}{\ell \choose k/2}} - \ell.$$
Set $\lambda = \lambda^* :=(1-\gamma) \rho m\delta$, $q = \frac{m}{{n \choose k}}$ for some $0<\gamma < 1$, $\delta^S = \frac{{n-\ell\choose k/2}{\ell\choose k/2}}{{n\choose k}}$, simplify and ignore the terms $O(n^{-k})$:
$$\vec{z}^\intercal V\vec{z} \gtrapprox n\frac{(1-\gamma)m\delta - \frac{t}{\rho}\sqrt{m\delta^S}}{m\delta^S} - \ell = n\left((1-\gamma)\frac{\delta}{\delta^S} - \frac{t}{\rho\sqrt{m\delta^S}}\right)-\ell.$$
Now we find that $\delta^S = (1-o(1)) {k \choose k/2}(\ell/(n-k/2))^{k/2}$, meaning that if $m = \Omega(n^{k/2}\ell \log n)$ the second term in the brackets is $O(\ell^{-k/4})$ and so vanishes at large $\ell$. Meanwhile $\delta/\delta^S \ge (1-O(1/n))(1-k^2/(4\ell))$ approaches $1$ at large $\ell$ and therefore $\frac{\vec{z}^\intercal V\vec{z}}{n} \ge 1 - \gamma- o(1) = \Omega(1)$.
\end{proof}

This ensures that with high probability we have weak recovery using the techniques in~\cite{hastings2020classical}. Briefly, this involves forming the matrix $\hat{V} = (\mathbb{1} + V)/n$. In the bosonic picture, this is the one-particle reduced density matrix (1RDM): indeed $\hat{V}$ is clearly trace-1, and it is also positive since, for any $\vec{a} \in \mathbb{C}^n$, $2\vec{a}^\dagger \hat{V}\vec{a} = \<v|\hat{X}^\dagger \hat{X}|v\> + \<v|\hat{Y}^\dagger \hat{Y}|v\> \ge 0$ with $\hat{X} = \sum_i a_iX_i,\,\hat{Y} = \sum_i a_iY_i$.
We now sample a Gaussian vector from $\mathcal{N}(0, \hat{V})$ and normalize it to form $\vec{x}$. With high probability, $|\vec{x}^\intercal \vec{z}| \ge c\sqrt{\vec{z}^\intercal\hat{V}\vec{z}}>c\sqrt{\frac{\vec{z}^\intercal V\vec{z}}{n}}$ for some constant $c$~\cite{hastings2020classical}.

\subsection{Boosting to strong recovery}
We want to prove that boosting works for the tensor completion case. Let us assume access to the entire tensor as opposed to only the entries with unique indices, let the new tensor be $T'$. To maintain compatibility with the previous proofs, we again assume a Skellam distribution for the tensor entries, such that
$$T'_s = \Xi_s z_s,\;\;\forall s \in [n]^k.$$
Note that now the Skellam variables are indexed by multi-indices $s$ with possibly repeated indices, and where the order of the indices does matter.
The model for the quantum algorithm we use is therefore: 
\begin{enumerate}
    \item With access to $T'$, form $T = \{T_S\}$ by selecting the entries with unique indices and summing the entries with the same indices up to permutation.
    \item Run the quantum algorithm with $T$ and obtain a candidate solution $\vec{x}$.
    \item Perform boosting with one round of tensor power iteration with $T'$.
\end{enumerate}
Accordingly, we let $\Xi_s \sim \mathrm{Skellam}(\frac{1+\rho}{2}q', \frac{1-\rho}{2}q')$ with $q' = \frac{m}{k!{n\choose k}}$ such that selecting the entries with unique indices and combining those with the same indices gives the correct distribution: $T_S = \Xi_S z_S$, $\Xi_S \sim \mathrm{Skellam}(\frac{1+\rho}{2}q, \frac{1-\rho}{2}q)$.

We are going to prove that tensor power iteration outputs a vector $\hat{x}$ with correlation with the planted spike $\vec{z}$ asymptotically $1 - o(1)$, therefore achieving strong recovery.

Assume from the candidate solution is normalized $\|\vec{x}\|=1$, and let $r = \vec{x}^\intercal \vec{z}$, which we assume to be positive. We then have:
\begin{proposition}\label{prop:boosting}
    With $m = \Omega(n^{k/2} \log n)$, if $r = \omega(n^{-1/4})$ then with high probability
    $$\mathrm{corr}(\hat{x},\vec{z}) \in 1 - o(1).$$
\end{proposition}
\begin{proof}
    
    We define the boosted solution as $\hat{x} =  T' \cdot \vec{x}^{\otimes (k-1)}$.
    Again we define the tensor $\xi$ with $\xi_s = \Xi_s - \rho q'$, a centered Skellam random variable with variance $q'$, such that $T' = \rho q'\vec{z}^{\otimes k} + \xi \odot\vec{z}^{\otimes k}$.
    We split $\hat{x}$ into signal and noise components:
    $$\hat{x} = \hat{x}^{(S)} + \hat{x}^{(N)},\;\;\;\hat{x}^{(S)} = \rho q' r^{k-1} \vec{z}, \;\;\;\hat{x}^{(N)}_i = z_i \sum_{\substack{s\in [n]^k \\s \ni i}} \xi_sz_{s/i}x_{s/i}.$$ 
    Therefore,
    $$\hat{x}^\intercal z \ge \rho q'n r^{k-1} - \|x^{(N)}\|\sqrt{n}$$
    and 
    $$\|\hat{x}\| \le \rho q' \sqrt{n}r^{k-1} + \|x^{(N)}\|$$
    so
    $$\text{corr}(\hat{x},\vec{z}) = \frac{\hat{x}^\intercal \vec{z}}{\|\hat{x}\|\sqrt{n}} \ge \frac{\rho q'\sqrt{n} r^{k-1} - \|x^{(N)}\|}{\rho q' \sqrt{n}r^{k-1} + \|x^{(N)}\|}.$$

    Using the Bernstein inequality in Prop.~\ref{prop:bernstein},
    $$\mathbb{P}(|\hat{x}^{(N)}_i| \ge t) \le 2e^{-t^2/4\sigma^2}\;\;\text{if } t \le \sigma^2,$$
    where 
    $$\sigma^2 = \sum_{\substack{s\in [n]^k \\s \ni i}} \mathrm{Var}\left(\xi_sz_{s/i}x_{s/i}\right) = q'\sum_{\substack{s\in [n]^k \\s \ni i}} (x_{s/i})^2.$$
    The monomials in the sum are a subset of those in $\|\vec{x}\|^{2(k-1)} = \left(\sum_i x_i^2\right)^{k-1}$, with the largest coefficient being $k!$. Therefore $\sigma^2  \le k!q'$.
    Using a union bound, with high probability, 
    $$\max_i |\hat{x}^{(N)}_i| = O(\sqrt{q' \log n}) \implies \|x^{(N)}\| = O(\sqrt{q'n \log n}).$$
    Therefore, as long as $\sqrt{m}r^{k-1} = \omega(\sqrt{\log n})$, the signal part will dominate in both the denominator and numerator. Since $k \ge 4$, if $m = \Omega(n^{k/2} \log n)$ it is sufficient to have $r = \omega(n^{-1/4})$.
\end{proof}

\section{Asymmetric tensors}\label{sec:asymm}
In this section we show how to extend the quartic advantage in~\cite{schmidhuber2024quartic} to general, asymmetric tensors in $(\mathbb{R}^{n})^{\otimes k}$. Once again, we focus on even $k$.

We first review some results from the paper on symmetric tensor embeddings~\cite{ragnarsson2013block}. Let $T\in \mathbb{R}^{n_1 \otimes \cdots\otimes n_k}$ be a generic (asymmetric) $k$-order tensor. Then let $N = n_1 + \cdots + n_k$, and consider tensors $\in (\mathbb{R}^{N})^{\otimes k}$. Impose a block structure to such tensors by splitting, for each dimension, the $N$ components along the various $n_i$s, writing $(i,j)$ for the $j$th component of the $i$th block, $1\le i \le k$ and $1 \le j \le n_i$. Therefore $(i,j)$ corresponds to the index $\sum_{k<i} n_k + j$. The symmetric embedding of $T$ is the tensor defined as 
\begin{equation}
    [\text{sym}(T)]_{(i_1,j_1), ...,(i_k,j_k)} := \begin{cases}
        T_{\pi^{-1}(j_1\cdots j_k)} &\text{if}\;\; [i_1\cdots j_k] = \pi([1\cdots k]),\\
        0 &\text{otherwise}
    \end{cases}
\end{equation}
where $\pi$ is a permutation. 
From now on we assume that $n_1=...=n_k = n$ so $N = kn$, however this assumption may be easily relaxed.

Since $\text{sym}(T)$ is symmetric, we can alternatively index it by unordered lists $I = \{(1,j_1), ...,(k,j_k)\}$ such that $[\text{sym}(T)]_I = T_{j_1 \cdots j_k}$. Let $\sigma$ be the map $I \mapsto (j_1...j_k)$. We call multi-indices of $\text{sym}(T)$ of this form, which correspond to entries of $T$, \textit{valid indices} of $\text{sym}(T)$.
We generalize this definition for later use: we let $\mathcal{V}_{N,\ell,k}$ be the set of all multi-indices of length $\ell=kc$ sequences from $[N]=[kn]$, such that there are exactly $c$ indices from each of the $k$ blocks. Therefore the valid indices of $\mathrm{sym}(T)$ correspond to $\mathcal{V}_{N,k,k}$.

This construction is a generalization of the symmetrization of a rectangular matrix
$$
A \in \mathbb{R}^{n_1\times n_2} \to \text{sym}(A)=\begin{pmatrix}
    0 & A\\
    A^T & 0
\end{pmatrix} \in \mathbb{R}^{(n_1 + n_2) \times (n_1 + n_2)}
$$
and similarly to the matrix case there is a simple mapping between the singular values of $T$ and the eigenvalues of $\text{sym}(T)$~\cite{ragnarsson2013block}.

\subsection{Kikuchi construction}
The tensor model now becomes: select entries by Poisson sampling with frequency $q = \frac{m}{n^k}$ of the tensor
$$T^* = \bigotimes_{i=1}^k \vec{z}_{i}$$
with $\{\vec{z}_{i}\}_{i=1}^k$ a sequence of vectors $\in \{\pm1\}^n$. For each sample, flip the sign with probability $\frac{1-\rho}{2}$, thus forming the tensor $T$. Once again, we have $T_s = T^*_s\Xi_s$ where $\{\Xi_s\}_{s \in [n]^k}$ are iid $\text{Skellam}(\frac{1+\rho}{2}q, \frac{1-\rho}{2}q)$ random variables. Notice that the tensor multi-indices now range over $[n]^k$, where the index order does matter.

Now symmetrically embed the tensor into $\mathrm{sym}(T)$. The construction of the Kikuchi matrix then proceeds identically to the symmetric tensor case. 
Forming the Kikuchi matrix removes the symmetry in the indices, because it is indexed by \textit{unordered} lists $U,V \in {[N] \choose \ell}$. Explicitly,
\begin{equation}
    \KC(\text{sym}(T)) = \sum_{S \in \mathcal{V}_{N,k,k}}  T_{\sigma(S)}A^{S}
\end{equation}
where $A^{S}$ is a matrix defined as
\begin{equation}
    [A^{S}]_{U,V} = \begin{cases}
        1 &\text{if } U \Delta V = S,\\
        0 &\text{otherwise.}
    \end{cases}
\end{equation}
In this space, given a vector $\vec{x} \in \mathbb{R}^N$ we can define $\vec{x}^{\odot \ell}$ to be the vector $\in \mathbb{R}^{N\choose \ell}$ such that $[\vec{x}^{\odot \ell}]_S = x_S$ and $|x^{\odot \ell}\>$ to be the corresponding quantum state (appropriately normalized). We also define
$$\vec{z} = \begin{bmatrix}
        \vec{z}_1\\ \vec{z}_2\\ \cdots\\ \vec{z}_k \end{bmatrix} \in \{1,-1\}^{N}.$$

\subsection{Detection for spiked tensor}
Here we show the existence of high-energy eigenstates for the Kikuchi matrix of a symmetrized planted asymmetric tensor.
Consider a generic asymmetric tensor $T\in (\mathbb{R}^{n})^{\otimes k}$, let $S = \text{sym}(T) \in \mathbb{R}^{N\times N}$, $N=nk$. Then clearly
$$\max_{\{\vec{x}_i\} \in \{\{\pm1\}^n\}^k} T\cdot\bigotimes_{i=1}^k \vec{x}_i = \frac{1}{k!} \max_{\vec{x} \in \{\pm1\}^N} S\cdot\vec{x}^{\otimes k}.$$
At the same time, using the identity for all $\ell$-order Kikuchi matrices:
$$S\cdot\vec{x}^{\otimes k} = k!\frac{(\vec{x}^{\odot\ell})^\intercal\KC(S)\vec{x}^{\odot \ell}}{{k \choose k/2}{N-k \choose \ell-k/2}},$$ we get that
$$\max_{\vec{x} \in \{\pm1\}^N} S\cdot\vec{x}^{\otimes k}  = k!\delta_{N,k,\ell} \max_{\vec{x} \in \{\pm1\}^N}\<x^{\odot \ell}|\KC(S)|x^{\odot \ell}\>\le k!\delta_{N,k,\ell}  \lambda_{\max}(\KC(S)).$$
Notice $\delta_{N,k,\ell} \sim k^{-k/2}\delta_{n,k,\ell}$. 

Next we shift to our noisy sparse tensor model, and use the contraction with the spike to lower bound these quantities. 
$$\max_{\{\vec{x}_i\} \in \{\{\pm1\}^n\}^k} T\cdot\bigotimes_{i=1}^k \vec{x}_i \ge T\cdot\bigotimes_{i=1}^k \vec{z}_i = \sum_{s \in [n]^k} T_S^{*2} \Xi_s = \sum_{s \in [n]^k} \Xi_s.$$
This quantity is itself a random variable $\sim \text{Skellam}(\frac{1+\rho}{2}m, \frac{1-\rho}{2}m)$, and therefore we can use Prop. 2.29 in~\cite{schmidhuber2024quartic} to claim that:
$$P\left(T\cdot\bigotimes_{i=1}^k \vec{u}_i \le (1-\gamma)\rho m\right) \le e^{-\frac{\gamma^2\rho^2}{2}m}.$$
Therefore, with high probability, 
$$ \lambda_{\max}(\KC(S)) \ge (1-o(1))\rho d_{N,k,\ell,m}.$$
Choosing $m \sim n^{k/2} \log n$, sets the threshold for detection at $\lambda^* = (1-\gamma)\rho d_{N,k,\ell,m}\sim \rho k^{-k/2}\ell^{k/2} \log n$. The difference with the symmetric case is therefore a factor of $ k^{-k/2}$.

\subsection{No detection for random tensor}
To keep consistency with prior works~\cite{wein2019kikuchi, schmidhuber2024quartic}, we first present a result for a Rademacher tensor model, where $T$ is composed of uniform random signs and not symmetric, with all but $m$ entries picket at random masked. The symmetrized tensor will be structured but we will see that with high probability the corresponding Kikuchi matrix $\KC_{\text{rand}}$ still has no large eigenvalues. We use the following standard bound:
\begin{theorem}[Matrix Chernoff bound~\cite{tropp2012user}]
    Let $\{M_i\}$ be fixed $d\times d$ matrices, and let $\{\xi_i\}$ be independent standard normal or Rademacher random variables. Let $\Sigma = \sum_i \xi_iM_i$ and $\tau^2 = \|\mathbb{E}\;\Sigma^2\|$. Then for all $t \ge 0$
    \begin{equation}
        \mathbb{P}(\|\Sigma\| \ge t) \le 2de^{-t^2/2\tau^2}.
    \end{equation}
\end{theorem}
In our case $d = n^k$, the $\xi_i$ are Rademacher random variables, and using the properties of the Kikuchi matrix we have that 
$$\tau^2 = \max_{U} D_{U,U} =: d_{\max}^S$$
where $D = \sum_{j_1...j_k} (A^{(j_1...j_k)})^2$ is the diagonal matrix of the degrees of the Kikuchi graph, and $D_{U,U}$ is the degree of node $U$. As in the symmetric case~\cite{schmidhuber2024quartic}, we fix the $U$ and consider $D_{U,U}$ to be the random variable obtained by choosing the multi-indices of the $m$ nonzero element of $T$ at random without replacement. We have that $D_{U,U} = \sum_{i=1}^m X^U_i$ where $X^U_i$ is a random sample from the list $(x^U_1,...,x^U_{n^k})$, where $x^U_i = 1(|U \Delta V_i| = \ell)$ and $V_i \in \mathcal{V}_{N,k,k}$ is the $i$th multi-index of $\mathrm{sym}(T)$. 

To bound $D_{U,U}$ with a multiplicative Chernoff bound we need to know $\mu_U = \mathbb{E}X^U_i=\frac{1}{n^k}\sum_i x^U_i$, which due to the structure of the symmetrized tensor may now depend on $U$. 
Specifically this quantity depends on the number of elements in $U$ which belong to the same block, and is maximized when $U \in \mathcal{V}_{N,\ell,k}$, in which case $\mu_U = \delta^S_{N,k,\ell}$ with
$$\delta^S_{N,k,\ell} := \frac{{k \choose k/2}c^{k/2} (n-c)^{k/2}}{n^k} = {k \choose k/2} \left(\frac{\ell}{N}\right)^{k/2}\left(\frac{N-\ell}{N}\right)^{k/2}.$$
Therefore over the randomness of the masking, $\mathbb{E} D_{U,U} = m\delta^S_{N,k,\ell} =: d^S_{N,k,\ell,m}$. For the other multi-indices, $\mu_U \le \delta^S_{N,k,\ell}$. Therefore the same concentration bound using $\mu_U=\delta^S_{N,k,\ell}$ applies uniformly to all multi-indices. We can then use the multiplicative Chernoff bound and a union bound over the $U$s and obtain that for all $\kappa >0$:
$$
    \mathbb{P}(d_{\max}^S \ge (1+\kappa)d^S_{N,k,\ell,m}) \le e^{\log {N\choose \ell} -\frac{\kappa^2}{2+\kappa}d^S_{N,k,\ell,m}}.
$$
We see that if we set $m \ge kn^{k/2}\log n$ then $ d^S_{N,k,\ell,m} \ge (1-\frac{\ell}{2n}){k \choose k/2} kc^{k/2} \log n$. Choosing $\kappa =1$, since $k \ge 4$ and $\ell \le n/2$ then $\frac{\kappa^2}{2+\kappa}d^S_{N,k,\ell,m} \ge 3 kc^{k/2}\log n \ge \frac{3}{2}\log {N \choose \ell}$. So
$$\mathbb{P}(d_{\max}^S \ge 2d^S_{N,k,\ell,m}) \le {N \choose \ell}^{-\frac{1}{2}}.$$
Therefore we get that, with $m = kn^{k/2}\log n$,
$$\|\KC_{\text{rand}}\| \le \sqrt{6{k \choose k/2}kc^{k/2}\log n \log{N \choose \ell}}\le \sqrt{12{k \choose k/2}}\ell^{(k+2)/4}k^{(2-k)/4}\log n$$ 
except with probability $\le 3{N \choose \ell}^{-1/2}$.
Comparing this with the bound for the spiked tensor, we see that at constant $k$ there is a gap $\lambda^*/\|\KC_{\text{rand}}\| = \Omega(\ell^{k/2})$ between the eigenvalue and therefore we have detection.
Notice that we have obtained an upper bound for the maximum degree of the Kikuchi graph $d^S_{\max}$ that will be of use later:
\begin{corollary}
    The maximum degree of the Kikuchi graph for a symmetrized tensor is
    $$d^S_{\max} \le (1+\kappa)d^S_{N,k,\ell,m}$$
    for all $\kappa >0$, except with probability at most $e^{\log {N\choose \ell} -\frac{\kappa^2}{2+\kappa}d^S_{N,k,\ell,m}}$.
\end{corollary}

We now repeat the proof for $T$ a random tensor with entries $T_s = \Xi_s \sim \text{Skellam}(\frac{q}{2}, \frac{q}{2})$, which corresponds to selecting $\rho = 0$ in the spiked model.
The proof is similar, with the difference that we use of the Bernstein inequality in Prop.~\ref{prop:bernstein} to write: if $t \le \sigma^2$,
$$\mathbb{P}(\|\KC_{\text{rand}}\|\ge t) \le 2{N\choose \ell}e^{-t^2/4\sigma^2},\;\;\;\sigma^2 =q\Delta^S_{\max}.$$
Here $\Delta^S_{\max}$ is the maximum degree of the Kikuchi graph of $\mathrm{sym}(T)$ for $T$ a dense (complete) tensor, which is $\Delta^S_{\max} = {k \choose k/2}c^{k/2} (n-c)^{k/2}$. Therefore $\sigma^2 =q\Delta^S_{\max} = d^S_{N,k,\ell,m}$.
We get that with $m = kn^{k/2}\log n$,
$$\|\KC_{\text{rand}}\| \le  \sqrt{8{k \choose k/2}}\ell^{(k+2)/4}k^{(2-k)/4}\log n$$
except with probability at least $2{N\choose \ell}^{-1}$. The result is very similar, and even slightly stronger, than the Rademacher case.

\subsection{Bounding the guiding state overlap}
The results so far show that the Kikuchi method is successful at detecting the planted spike. The final ingredient for quantum advantage lies in showing the existence of a guiding state with nontrivial support on the large eigenvalue subspace of $\KC(\text{sym}(T))$. The proof as presented in~\cite{schmidhuber2024quartic} is comprised of the following two steps: 1) the state $|z^{\odot \ell}\>$ corresponding to the planted spike has a large support on a subspace of high energy, 2) any vector $|v\>$ has an overlap with the guiding state that is lower bounded by the overlap with $|z^{\odot \ell}\>$. The result follows from choosing $|v\> \propto \Pi_{>\lambda^*} |z^{\odot \ell}\>$ and lower bounding its overlap with guiding state. 

In the asymmetric case, the state $|z^{\odot \ell}\>$ corresponds to the planted string $\vec{z}$ defined above.
One can show that this specific quantum state has a large expectation value. 
However, we have an issue using this method: the guiding state that we can prepare only has support on the valid indices $\mathcal{V}_{N,\ell,k}$. Indeed, we can generate a state parallel to $|\Gamma\>$ by taking $|\phi\>^{\otimes c}$ and projecting to the subspace of Hamming weight $\ell$, where $|\phi\> \propto \sum_{U \in \mathcal{V}_{N,k,k}} \mathrm{sym}(T)_{U}|U\>$. One can check that we indeed only get support on multi-indices in $\mathcal{V}_{N,\ell,k}$.

Define the unnormalized state
$$|\Gamma\> = \frac{1}{\chi} \sum_{U \in \mathcal{V}_{N,\ell,k}} x_U H^{\star c}_U |U\>, \;\;\; \chi := {n \choose c}^{k/2} (c!)^{(k-1)/2}$$
$$H^{\star c}_U := \frac{1}{q^{c/2}} \sum_{\mathcal{T} \in \text{Part}^v_k(U)} \prod_{T \in \mathcal{T}} x_TS_T ,\;\;\; q := \frac{m}{n^k}$$
where $\text{Part}^v_k(U)$ is the set of valid $k$-partitions of $U$, $\text{Part}^v_k(U) = \text{Part}_k(U) \cap (\mathcal{V}_{N,k,k})^c$. Note that if $U \notin \mathcal{V}_{N,\ell,k}$, $\text{Part}^v_k(U) = \emptyset$. The normalization of $H^{\star c}_U$ is because $x_TS_T \sim \text{Skellam}(\frac{1+\rho}{2} q, \frac{1-\rho}{2} q)$ and so $H_T := \frac{1}{q^{1/2}} x_TS_T$ has mean $\rho\sqrt{q}$ and variance 1.
Also define the normalized state
$$|\tilde{z}\> = {n \choose c}^{-k/2} \sum_{U \in \mathcal{V}_{N,\ell,k}} z_U |U\>.$$
We then have the following.
\begin{proposition}
Over the randomness of the signs and the observed tensor entries,
$$\mathbb{E}\left[|\Gamma\>\right] = \rho^c \frac{m^{c/2}}{n^{\ell/2}}(c!)^{(k-1)/2} |\tilde{z}\>.$$
\end{proposition}
\begin{proof} 
    Note that
    $$\mathbb{E} H^{\star c}_U = \frac{1}{q^{c/2}} \sum_{\mathcal{T} \in \text{Part}^v_k(U)} \prod_{T \in \mathcal{T}} \mathbb{E} (x_TS_T) = \rho^c q^{c/2} |\text{Part}^v_k(U)|.$$
    Therefore
    $$\mathbb{E}|\Gamma\> = \frac{\rho^c q^{c/2}}{\chi} \sum_{U \in {[N] \choose \ell}} x_U|\text{Part}^v_k(U)||U\>.$$
    For valid $U$s, $|\text{Part}^v_k(U)| = (c!)^{k-1}$ (since we fix the ordering of the $c$ indices in the first block and then independently choose an ordering for each of the remaining $k-1$ blocks). Overall, we have
    $$\mathbb{E}|\Gamma\> = \frac{\rho^c q^{c/2} (c!)^{k-1}}{\chi} \sum_{U \in \mathcal{V}_{N,\ell,k}} z_U|U\> = \rho^c \frac{m^{c/2}}{n^{\ell/2}}(c!)^{(k-1)/2} |\tilde{z}\>,$$
    completing the proof.
\end{proof}
We see that if we set $m \sim n^{k/2}$ we indeed have $\mathbb{E} \<\Gamma|\tilde{z}\> \sim n^{-\ell/4}$ as in the symmetric case, suggesting that the quartic advantage is retained, if we can prove a lower bound for $\<\tilde{z}|\Pi_{> \lambda^*}|\tilde{z}\>$. However, one finds that in fact $\<\tilde{z}|\KC(\mathrm{sym}(T))|\tilde{z}\> = 0$, and therefore the previous techniques for showing a large $\<\tilde{z}|\Pi_{> \lambda^*}|\tilde{z}\>$ fail.

The solution is to use the fact that even though $|\tilde{z}\>$ has expectation value zero, it has a large variance and therefore we can bound its mass on the high energy subspace. We use the following universal result:
\begin{proposition}
    Suppose $X$ is a random variable that has $\mathbb{E}\,X=0$, $\mathrm{Var}\,X  = \sigma^2$, and is bounded $|X| \le M$ a.s.. Then for all $\theta \ge 0$:
    $$P\left(X > \theta\frac{\sigma^2}{M}\right) \ge \begin{cases}
        (1-\theta)\frac{M^2\sigma^2 + \theta\sigma^4}{2(M^4 - \sigma^4)} &\text{if } \theta < 1,\\
        0 &\text{otherwise.}
    \end{cases}$$
    The bound is tight.
\end{proposition}
\begin{proof}
    Consider $\mathbb{P}(X > t)$ for $0 \le t \le M$.
    When $t \ge \frac{\sigma^2}{M}$, there exist distributions that obey the constraints and have zero mass beyond $t$. For instance, consider the two-atom distribution $\mathbb{P}(X = \frac{\sigma^2}{M}) =\epsilon$, $\mathbb{P}(X=-M) = 1-\epsilon$, with $\epsilon = \frac{M^2}{\sigma^2 + M^2}$.

    When $t < \frac{\sigma^2}{M}$, no such distribution exist. 
    Note that actually we are solving a generalized moment problem. The classical theory (see, e.g., \cite{bertsimas2005optimal}) states that the extremal value of a linear functional under $k$ moment constraints and bounded support is always achieved by a discrete distribution supported on at most $k+1$ points. In our case, with three constraints (mean, variance, and total probability), the extremal distribution is supported on at most three points. This follows from the Krein–Milman theorem and is a standard result in the theory of moment problems. 
    In this case, the extremal distribution has three atoms: $\mathbb{P}(X = M) = \eta$, $\mathbb{P}(X = t) = \epsilon$, $\mathbb{P}(X=-M) = 1-\epsilon-\eta$.
    Solving for the constraints gives
    $$\epsilon = \frac{M^2 - \sigma^2}{M^2 - t^2},\;\;\eta= \frac{\sigma^2(1+\frac{t}{M}) - t^2(1+\frac{M}{t})}{2(M^2 - t^2)}.$$
    Now reparameterizing in terms of $\theta = \frac{tM}{\sigma^2}$ gives the result.
\end{proof}
From this we get the following:
\begin{corollary}\label{cor:1}
    If $\lambda^* = \frac{\<\KC^2\>}{ 2d_{max}^S}$,
    $$\<\tilde{z}|\Pi_{> \lambda^*}|\tilde{z}\> \ge \frac{\<\KC^2\>}{4d_{max}^{S2}}.$$
\end{corollary}
\begin{proof}
    Setting $\theta = 1/2$ in the proposition we can use the bound 
    $P\left(X > \frac{\sigma^2}{2M}\right) \ge \frac{\sigma^2}{4M^2}
    $.
    We let $X$ to have the distribution on the spectrum of $\KC$ induced by $|\tilde{z}\>$. We get the result by plugging in $\sigma^2 =\<\KC^2\> = \<\tilde{z}|\KC^2|\tilde{z}\>$, $M= d_{max}^S$ and using $\mathbb{P}(X > \lambda^*) = \<\tilde{z}|\Pi_{> \lambda^*}|\tilde{z}\>$ with $\lambda^* = \frac{\<\KC^2\>}{ 2d_{max}^S}$.
\end{proof}

We can therefore obtain the following proposition.
\begin{proposition}
    For $0 < \gamma < 1$,
    $$\<\KC^2\> \ge (1-\gamma)^2 c^k\rho^2 \frac{m^2}{n^k}$$
    except with probability at most ${n \choose c}^k e^{-\frac{\gamma^2\rho^2}{2}m}$.
\end{proposition}
\begin{proof}
    Note that
    \begin{align}
    \<\KC^2\> = \|\KC|\tilde{z}\>\|^2
    &= \frac{1}{{n \choose c}^k} \sum_{V \in {N\choose \ell}}\left(\sum_{U \in \mathcal{V}_{N,\ell,k}} \text{sym}(T)_{U\Delta V}\,z_U,\right)^2 \\
    &=\frac{1}{{n \choose c}^k} \sum_{V \in {N\choose \ell}}\left(\sum_{S\in\SC} \Xi_S \,1(S\Delta V \in \mathcal{V}_{N,\ell,k})\right)^2.
    \end{align}
    Each term of the sum over $T$ is a Skellam-distributed random variable $$\sum_{S\in\SC} \Xi_S \,1(S\Delta V \in \mathcal{V}_{N,\ell,k}) \sim \text{Skellam}(\frac{1+\rho}{2} m_V, \frac{1-\rho}{2} m_V),$$ with $$m_V = \mathbb{E}_{\SC}\left|\{S \in \SC\,|\,S\Delta V \in \mathcal{V}_{N,\ell,k}\}\right| = \begin{cases}
        m \frac{(n-c+1)^{k/2}(c+1)^{k/2}}{n^k} &\text{ if } \exists S \in \mathcal{V}_{N,k,k} \text{ s.t. } S\Delta V\in \mathcal{V}_{N,\ell,k},\\
        0 &\text{ otherwise.}
    \end{cases}$$
    So for the nonzero $m_V$'s we have that  $\sum_{S\in\SC} \Xi_S 1(S\Delta V \in \mathcal{V}_{N,\ell,k}) \ge (1-\gamma)\rho m_V$ except with probability at most $e^{-\frac{\gamma^2\rho^2}{2}m}$, for $0 < \gamma < 1$. Now using a union bound we have that
    \begin{align*}
    \<\KC^2\> &\ge (1-\gamma)^2\rho^2 m^2 \frac{{n \choose c-1}^{k/2}{n \choose c+1}^{k/2}}{{n \choose c}^k} \frac{(n-c+1)^{k}(c+1)^{k}}{n^{2k}}\\
    &\ge (1-\gamma)^2 c^k\rho^2 \frac{m^2}{n^k} 
    \end{align*}
    except with probability at most ${n \choose c-1}^{k/2}{n \choose c+1}^{k/2} e^{-\frac{\gamma^2\rho^2}{2}m} < {n \choose c}^{k} e^{-\frac{\gamma^2\rho^2}{2}m}$. 
\end{proof}
Recall that $d_{max}^S \le (1+ \kappa)d^S$ except with probability at most ${N \choose \ell} e^{-\frac{\kappa^2}{2 + \kappa}d^S}$, with $d^S \approx {k \choose k/2}m c^{k/2}n^{-k/2}$. We have our final bound on the size of the support of $|\tilde{z}\>$ on the high energy subspace:
\begin{proposition}
    With
    $$\lambda^* = \frac{(1-\gamma)^2\rho^2}{2(1+\kappa){k\choose k/2}}c^{k/2} m n^{-k/2},$$
    we have
    $$\|\Pi_{> \lambda^*}|\tilde{z}\>\| \ge \frac{(1-\gamma)\rho}{2(1+\kappa){{k \choose k/2}}}$$
    except with probability at most ${N \choose \ell} e^{-\frac{\kappa^2}{2 + \kappa}d^S} + {n \choose c}^{k} e^{-\frac{\gamma^2\rho^2}{2}m}$.
\end{proposition}

Now it remains to show that $|\Gamma\>$ concentrates:
\begin{proposition}
For any unit vector $|v\>$, provided that $\rho^2q \le \min\left\{\frac{1}{100c},  \frac{k-1}{2(c-1)(c+1)^{k-2}}\right\}$,
    $$\mathrm{Var}\,\<v|\Gamma\> \le 2.04 (\rho \sqrt{q})^{2c-2}\frac{c^{2}(c!)^{k-1} }{(n-c+1)^k},$$
    $$\mathbb{E}\,\<\Gamma|\Gamma\> \le 2.04.$$
\end{proposition}
\begin{proof}
    We proceed in an analogous manner to the proofs in~\cite{schmidhuber2024quartic}. We use their Lemma 2.37: assuming $\rho^2q \le \frac{1}{100c}$, for $\{H_S\}_S$ iid distributed like $\frac{1}{q^{1/2}}\text{Skellam}(\frac{1+\rho}{2} q, \frac{1-\rho}{2} q)$, consider $H_A = H_{S_1}\cdots H_{S_c}$ and $H_B = H_{S'_1}\cdots H_{S'_c}$, such that exactly $a$ variables are shared between the two. Then
    $$
    \text{Cov}\,H_AH_B \begin{cases}
        =0 &\text{if } a=0,\\
        \le \mathbb{E}\, H_AH_B\le  e^{0.01} (\rho \sqrt{q})^{2(c-a)} &\text{otherwise.}
    \end{cases}
    $$
    Then if $|v\> = \sum_{U \in \mathcal{V}_{N,\ell,k}} w_U z_U |U\> + \sum_{U \notin \mathcal{V}_{N,\ell,k}} w_U |U\>$ with $\sum_{U \in {[N] \choose \ell}} w_U^2 =1$,
    \begin{align*}
        \mathrm{Var}\,\<v|\Gamma\> 
        &\le \frac{1}{{n \choose c}^k(c!)^{k-1}} \sum_{V\in \mathcal{V}_{N,\ell,k}} w_V^2 \sum_{U\in \mathcal{V}_{N,\ell,k}} |\mathrm{Cov}\, H_V^{\star c}H_U^{\star c}|\\
        \sum_{U\in \mathcal{V}_{N,\ell,k}} |\mathrm{Cov}\, H_V^{\star c}H_U^{\star c}|
        &\le \sum_{\mathcal{V} \in \text{Part}^v_k(V)} \sum_{U\in \mathcal{V}_{N,\ell,k}}\sum_{\mathcal{U} \in \text{Part}^v_k(U)}|\mathrm{Cov}\, H_{\mathcal{U}_1}\cdots H_{\mathcal{U}_c}H_{\mathcal{V}_1}\cdots H_{\mathcal{V}_c}| \\
        &\le |\text{Part}^v_k(V)| \sum_{\mathcal{U} \in \text{Part}^v_k(\ell)}|\mathrm{Cov}\, H_{\mathcal{U}_1}\cdots H_{\mathcal{U}_c}H_{\mathcal{V}_1}\cdots H_{\mathcal{V}_c}|\; (\text{fixed }\mathcal{V}\in \text{Part}^v_k(V))
    \end{align*}
    where $\text{Part}^v_k(\ell) = \bigcup_{U \in {[N]\choose \ell}} \text{Part}^v_k(U)$ is the set of all valid partitions of all the elements of ${[N]\choose \ell}$. Stratify the summation based on the number $a$ of $H$ variables that are shared, of which there are $f(a) = {c\choose a}{n \choose c-a}^k((c-a)!)^{k-1}$. Then for a fixed $\mathcal{V}\in \text{Part}^v_k(V)$,
    \begin{align*}
        \sum_{\mathcal{U} \in \text{Part}^v_k(\ell)} |\mathrm{Cov}\, H_{\mathcal{U}_1}\cdots H_{\mathcal{U}_c}H_{\mathcal{V}_1}\cdots H_{\mathcal{V}_c}|
        &\le e^{0.01} (\rho \sqrt{q})^{2c} \sum_{a=1}^c (\rho \sqrt{q})^{-2a}f(a)\\
        &\le 2e^{0.01} (\rho \sqrt{q})^{2c-2}f(1) \\
        &=2e^{0.01} (\rho \sqrt{q})^{2c-2}c{n \choose c-1}^k((c-1)!)^{k-1}
    \end{align*}
    So 
    $$\mathrm{Var}\,\<v|\Gamma\> \le 2e^{0.01} (\rho \sqrt{q})^{2c-2}\frac{c^{2}(c!)^{k-1} }{(n-c+1)^k}.$$
    Similarly, 
    \begin{align*}
        \mathbb{E} \,\<\Gamma|\Gamma\> &= \frac{1}{{n \choose c}^k(c!)^{k-1}} \sum_{U \in \mathcal{V}_{N,\ell,k}} \mathbb{E}\, (H_U^{\star c})^2 \\
        &= \frac{1}{{n \choose c}^k} \sum_{U \in \mathcal{V}_{N,\ell,k}}  \frac{1}{(c!)^{k-1}} \sum_{\mathcal{U}, \mathcal{U}' \in \text{Part}^v_k(U)}\mathbb{E}\,[ H_{\mathcal{U}_1}\cdots H_{\mathcal{U}_c} H_{\mathcal{U}'_1}\cdots H_{\mathcal{U}'_c}]\\
        \frac{1}{(c!)^{k-1}} \sum_{\mathcal{U}, \mathcal{U}' \in \text{Part}^v_k(U)}&\mathbb{E}\, [H_{\mathcal{U}_1}\cdots H_{\mathcal{U}_c} H_{\mathcal{U}'_1}\cdots H_{\mathcal{U}'_c}] =\sum_{\mathcal{U}' \in \text{Part}^v_k(U)}\mathbb{E}\, [H_{\mathcal{U}_1}\cdots H_{\mathcal{U}_c} H_{\mathcal{U}'_1}\cdots H_{\mathcal{U}'_c}] \; (\text{fixed }\mathcal{U}\in \text{Part}^v_k(U)).
    \end{align*}
    Stratify like before and define the number of choices for each $a$ as $g(a) = {c\choose a}((c-a)!)^{k-1}$. We see that provided that $\rho^2q < \frac{k-1}{2(c-1)(c+1)^{k-2}}$, $g(a+1)/g(a) \ge (\rho\sqrt{q})^{-2}\frac{k-1}{c-1}\left(\frac{k-1}{(c+1)(k-2)}\right)^{k-2} \ge 2$. Therefore
    \begin{align*}
        \sum_{\mathcal{U}' \in \text{Part}^v_k(U)}\mathbb{E}\, [H_{\mathcal{U}_1}\cdots H_{\mathcal{U}_c} H_{\mathcal{U}'_1}\cdots H_{\mathcal{U}'_c}]
        &\le e^{0.01} (\rho \sqrt{q})^{2c} \sum_{a=1}^c (\rho \sqrt{q})^{-2a}g(a) \\
        &\le 2 e^{0.01} g(c)  \le 2e^{0.01} 
    \end{align*}
    By using symmetry over $U \in \mathcal{V}_{N,\ell,k}$, we get the intended result.
\end{proof}

Now we can derive the final result:
\begin{theorem}
    Let $\gamma,\zeta,\nu$ be small constants $\in (0,1)$, and assume $\zeta\rho^2 \frac{m}{n^k} \le \min\left\{\frac{1}{100c},  \frac{k-1}{2(c-1)(c+1)^{k-2}}\right\}$. Then for all $\kappa$
    $$\frac{\<\tilde{z}|\Pi_{>\lambda^*}|\Gamma\>}{\|\Pi_{>\lambda^*}|\tilde{z}\>\|\||\Gamma\>\|} \ge \frac{(1-\gamma)\sqrt{\nu}}{5.72(1+\kappa)}\frac{\rho^{c+1}(c!)^{(k-1)/2}\zeta^{c/2}}{{k \choose k/2}} \frac{m_\Gamma^{c/2}}{n^{\ell/2}}$$
    except with probability
    $$\le \frac{33.28(1+\kappa)^2}{(1-\gamma)^2\zeta m\rho^4 c^{k-2}} \frac{1}{1-\frac{\ell-k}{n}} + {N\choose \ell}e^{-\frac{\kappa^2}{2+\kappa}(1-\zeta)d^S_{N,k,\ell,m}} +{n \choose c}^k e^{-\frac{\gamma^2\rho^2}{2}(1-\zeta)m} +\nu.$$
    The probability of failure is $O(1)$ and arbitrarily small for an appropriate choice of constants $\gamma,\zeta,\nu$ and $\frac{\kappa m}{n^{k/2}} = O(\log n)$.
\end{theorem}
\begin{proof}
    As in~\cite{schmidhuber2024quartic}, the dependence between $\KC$ and $|\Gamma\>$ poses a problem for the proofs. Therefore we similarly resort to Poisson splitting: we can prepare the state $|\Gamma\>$ and the Kikuchi matrix $\KC$ using different entries, with frequencies $m_\Gamma = \zeta m$ and $m_{\KC} = (1-\zeta)m$, for some $0 < \zeta <1$ (a small constant).
    
    By using Chebyshev's inequality, except with probability $\le \frac{4\mathrm{Var}\,\<v|\Gamma\>}{(\mathbb{E} \, \<v|\Gamma\>)^2}$, over the randomness of $|\Gamma\>$
    $$\<v|\Gamma\> \ge \frac{1}{2}\mathbb{E}_{\Gamma} \, \<v|\Gamma\>$$
    Fix $|v\> = \frac{\Pi_{>\lambda^*}|\tilde{z}\>}{\|\Pi_{>\lambda^*}|\tilde{z}\>\|}$. Then,
    $$\mathbb{E}_{\Gamma} \, \<v|\Gamma\> = \rho^c \frac{m_\Gamma^{c/2}}{n^{\ell/2}}(c!)^{(k-1)/2} \<v|\tilde{z}\> = \rho^c \frac{m_\Gamma^{c/2}}{n^{\ell/2}}(c!)^{(k-1)/2} \frac{\<\tilde{z}|\Pi_{> \lambda^*}|\tilde{z}\>}{\|\Pi_{>\lambda^*}|\tilde{z}\>\|} = \rho^c \frac{m_\Gamma^{c/2}}{n^{\ell/2}}(c!)^{(k-1)/2} \|\Pi_{>\lambda^*}|\tilde{z}\>\|.$$
    With probability $\ge 1-{N\choose \ell}e^{-\frac{\kappa^2}{2+\kappa}d^S_{N,k,\ell,m}} -{n \choose c}^k e^{-\frac{\gamma^2\rho^2}{2}m_{\KC}} $ over the randomness of $\KC$, 
    $$\mathbb{E}_{\Gamma} \, \<v|\Gamma\> \ge \frac{1-\gamma}{2(1+\kappa)}\frac{\rho^{c+1}(c!)^{(k-1)/2}}{{k \choose k/2}} \frac{m_\Gamma^{c/2}}{n^{\ell/2}}.$$
    Finally we also have the deterministic bound
    $$\mathrm{Var}\,\<v|\Gamma\> \le 2.04 (\rho \sqrt{q})^{2c-2}\frac{c^{2}(c!)^{k-1} }{(n-c+1)^k}$$ 
    provided that $\rho^2q \le \min\left\{\frac{1}{100c},  \frac{k-1}{2(c-1)(c+1)^{k-2}}\right\}$. Also by Markov's inequality
    $$\<\Gamma|\Gamma\> \le \frac{2.04}{\nu}$$
    except with probability at most $\nu$. Therefore overall
    $$\frac{\<v|\Gamma\>}{\||\Gamma\>\|} \ge \frac{(1-\gamma)\sqrt{\nu}}{5.72(1+\kappa)}\frac{\rho^{c+1}(c!)^{(k-1)/2}}{{k \choose k/2}} \frac{m_\Gamma^{c/2}}{n^{\ell/2}}$$
    except with probability
    $$\le \frac{33.28(1+\kappa)^2}{(1-\gamma)^2m_\Gamma\rho^4 c^{k-2}} \frac{1}{1-\frac{\ell-k}{n}} +{N\choose \ell}e^{-\frac{\kappa^2}{2+\kappa}d^S_{N,k,\ell,m}} + {n \choose c}^k e^{-\frac{\gamma^2\rho^2}{2}m_{\KC}} +\nu.$$
\end{proof}

\subsection{Overall quantum method for detection}
Overall, if $m = \Omega(n^{k/2}\log n)$, there is an energy gap of $\Omega(\ell^{k/2} \log n)$ between a random asymmetric tensor and a planted one. Therefore detection can be achieved by setting the threshold a constant factor larger than the spectral norm upper bound for a random tensor, which is $O(\ell \log n)$.
Meanwhile, we have shown that we can prepare a guiding state with support on a high-energy subspace with energy $\ge \lambda^*:=\frac{(1-\hat\gamma)^2\rho^2}{2(1+\kappa){k \choose k/2}}(\ell/k)^{k/2}  m n^{-k/2}$, except with probability at most ${N \choose \ell} e^{-\frac{\kappa^2}{2 + \kappa}d^S} + {n \choose c}^{k} e^{-\frac{\hat\gamma^2\rho^2}{2}m}$, for some $0 < \hat\gamma,\kappa<1$. We can choose $\hat\gamma$ and $\kappa$ constant to have a high probability of success. This gives us support of a subspace with energy $\sim \ell^{k/2} \log n$ when $m = \Omega(n^{k/2} \log n)$. Therefore we achieve detection.

Finally, the guiding state with high probability has support
$$\<\Gamma|\Pi_{>\lambda^*}|\Gamma\> \ge \xi\left(\frac{m}{n^{k}}\right)^{c}$$
with $\xi \sim\rho^{2c+2}(c!)^{k-1}{k \choose k/2}^{-2}$. Therefore with our setting of $m$ we have asymptotic quartic advantage for detecting planted asymmetric tensors.

\section{Numerics for recovery}\label{sec:numerics}

\begin{figure}
    \centering
    \includegraphics[width=0.48\linewidth]{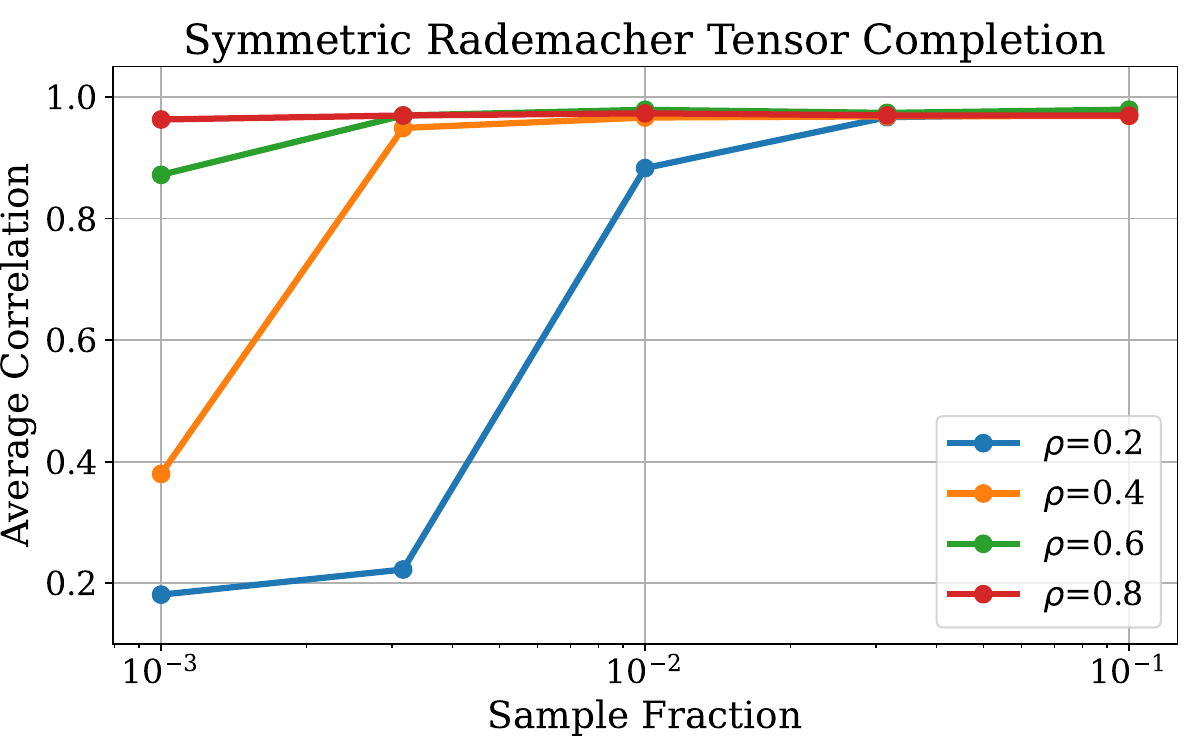}
    \includegraphics[width=0.48\linewidth]{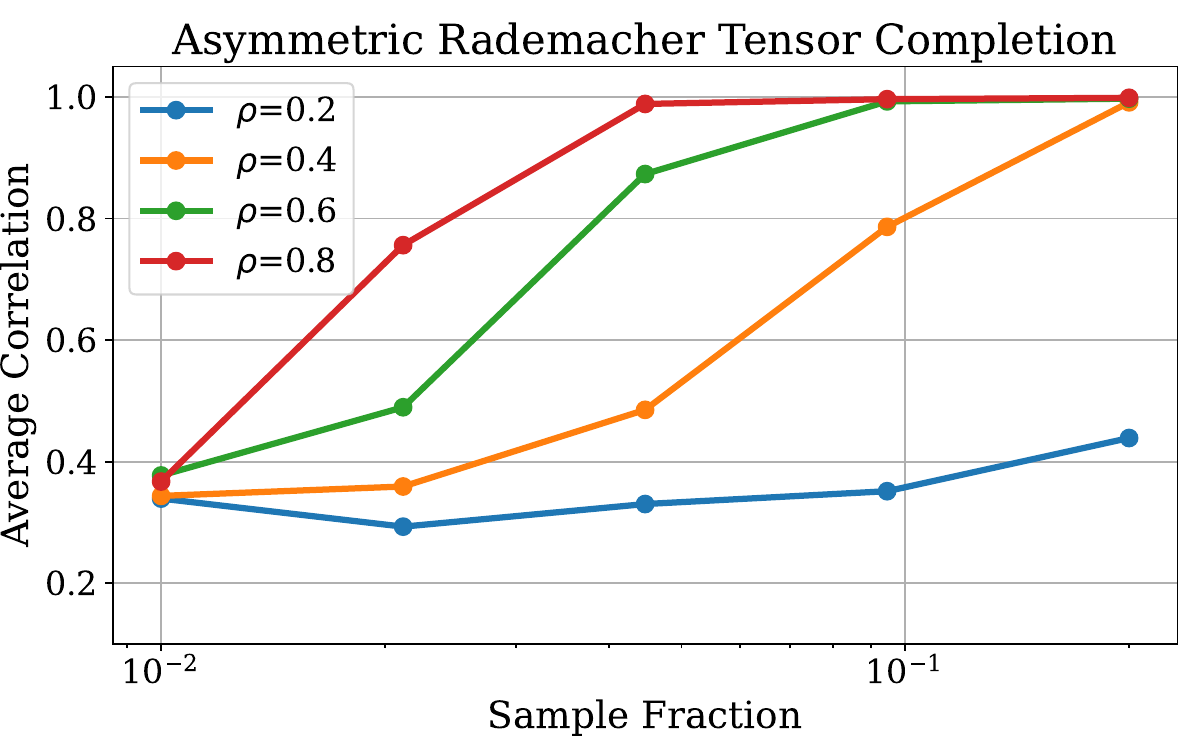}
    \caption{
    \textit{Performance of the Kikuchi method for symmetric and asymmetric tensor completion tasks $(k=4)$.} The setting is detailed in Section~\ref{sec:planted-kxor}. Reported points are the average of 30 independent trials. For the symmetric case, we choose $n=20$. For the asymmetric case, we choose $n=7$; due to symmetric embedding, this amounts to 28-dimensional tensor. In both cases, we use the Kikuchi method with $\ell=6$.
    }
    \label{fig:numerics}
\end{figure}

For asymmetric tensors, at the moment we are not able to provide a proof of recovery, since we must relate high-energy eigenvectors of two Kikuchi matrices that differ by a deterministic masking as opposed to a random masking in the symmetric case. However, we are confident that this can be shown to hold. In fact, we provide numerical evidence demonstrating that the Kikuchi method works for asymmetric tensor, in Fig.~\ref{fig:numerics}. 

The numerical experiments were also done for the symmetric case for comparison. The recovery method we use is slightly different from the one assumed in the theory. Briefly, we sample random signed tensors with a planted spike and noise determined by $\rho$, and we randomly mask a fraction of the entries. Then we construct the Kikuchi matrix and obtain its top 3 eigenvectors. Using a random linear combination of these eigenvectors we construct a corresponding voting matrix $V$. The proposed solution is taken as the topmost eigenvalue of $V$.

We observe that recovery of the spike is successful for both symmetric and asymmetric tensors. The performance of the Kikuchi method on symmetric tensors is particularly striking, since it reliably recovers the spike even at large noise (low $\rho$) and observation ratios lower than $1\%$. This provides direct evidence for the power of the Kikuchi method, and therefore the quantum algorithms presented in this work, in the setting of noisy low-rank tensor PCA and tensor completion.

\section{Conclusion}
We have presented the first end-to-end quantum algorithm for tensor problems, based on a method that is expected to achieve superquadratic speedups over classical methods. By introducing native qubit-based encodings and novel circuit optimizations, we reduced constant overheads by factors exceeding $10^4$, bringing quantum advantage for tensor PCA and completion closer to the reach of fault-tolerant quantum computers. Our algorithmic extensions to sparse recovery and asymmetric tensors, combined with explicit resource estimates showing $\sim 10^2-10^3$ logical qubits for meaningful problem instances, position tensor problems as among the most promising near-term applications for practical quantum advantage. These results suggest that, beyond the well-studied exponential speedups, polynomial quantum advantages may be both theoretically rigorous and practically viable.

\section{Acknowledgements}
The authors thank Rob Otter and Shaohan Hu for their support, and invaluable feedback on this project. They also thank their colleagues at the Global Technology Applied Research Center of JPMorganChase, and in particular Atithi Acharya, Javier Lopez-Piqueres, Niraj Kumar, and Yue Sun, for many helpful technical discussions.

\section*{Disclaimer}
This paper was prepared for informational purposes by the Global Technology Applied Research center of JPMorgan Chase \& Co. This paper is not a product of the Research Department of JPMorgan Chase \& Co. or its affiliates. Neither JPMorgan Chase \& Co. nor any of its affiliates makes any explicit or implied representation or warranty and none of them accept any liability in connection with this paper, including, without limitation, with respect to the completeness, accuracy, or reliability of the information contained herein and the potential legal, compliance, tax, or accounting effects thereof. This document is not intended as investment research or investment advice, or as a recommendation, offer, or solicitation for the purchase or sale of any security, financial instrument, financial product or service, or to be used in any way for evaluating the merits of participating in any transaction.

\bibliography{main}
\bibliographystyle{plain}

\appendix

\section{Appendix: Circuit construction}
\label{app:circuit}
\subsection{Guiding state preparation}
\label{subsec:guiding_state_preperation}
One of the steps in constructing the guiding state involves implementing a unitary operator $U$ with the following properties:
\begin{align}
    &U\ket{0}^{\otimes k} = \ket{0}^{\otimes k}\equiv \ket{0}^k,  \nonumber \\
    &U\ket{D_{1}^{(k)}} = \ket{1}\ket{0}^{\otimes k-1}\equiv \ket{e^{k}_1}, \nonumber \\
    &U\ket{D_{p}^{(k)}} = \ket{\psi}, \quad \text{for } 1 < p \leq k,
    \label{eq:condition_dicke_state}
\end{align}
where $\ket{D_{p}^{(k)}}$ denotes the Dicke state with $p$ excitations among $k$ qubits.
Furthermore, let $\rho' = \Tr_1\left( \ketbra{\psi}{\psi} \right)$ denote the partial trace over the first qubit. The following condition must also be satisfied:
\begin{align}
    {}^{k-1}\bra{0} \, \rho' \, \ket{0}^{k-1} = 0.
\end{align}

\begin{algorithm}
\caption{Dicke State Preparation Unitary ($U^{\dagger}$) }
\begin{algorithmic}[1]
\REQUIRE Integer $l$
\ENSURE Quantum Circuit

\STATE Initialize an empty quantum circuit
\STATE Define qubits as a linear array of size $2^l$

\FOR{$i = 1$ to $l$}
    \STATE Apply a controlled Hadamard gate on qubit at position $2^i-1$, controlled by the first qubit
\ENDFOR

\FOR{$i = 1$ to $l$}
    \IF{$i == 1$}
        \STATE Apply a CNOT gate between the second qubit and the first qubit
    \ELSE
        \FOR{$j = 0$ to $2^{i-1}-1$}
            \STATE Apply a Toffoli gate with control qubits at positions $2^i-1$ and $j$, and target qubit at $2^{i-1}+j$
            \STATE Apply a CNOT gate between qubits at positions $2^{i-1}+j$ and $j$
            \STATE Apply a CNOT gate between qubits at positions $2^{i-1}+j$ and $2^i-1$
        \ENDFOR
        \STATE Apply a CNOT gate between qubits at positions $2^i-1$ and $2^{i-1}-1$
    \ENDIF
\ENDFOR

\RETURN the constructed quantum circuit

\end{algorithmic}
\label{alg:circuit_dicke_state}
\end{algorithm}

As an example, we have the circuit that satisfies the condition in \cref{eq:condition_dicke_state} for $k=4$ given in \cref{fig:dicke_state_k_2}.

\begin{figure}[!ht]
    \centering
    \includegraphics[width=0.5\linewidth]{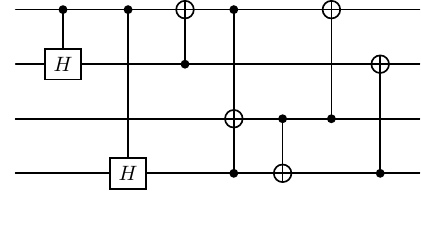}
    \caption{The circuit for $l=2$ ($k=2^l$) in the \cref{alg:circuit_dicke_state} which generates a circuit in \cref{eq:condition_dicke_state}. }
    \label{fig:dicke_state_k_2}
\end{figure}

\begin{proof}

Next, we prove that the circuit described in \cref{alg:circuit_dicke_state} generates a state that satisfies the condition in \cref{eq:condition_dicke_state}. To do so, we analyze each condition separately.

\textbf{Base Cases:}

First, consider the initial state $\ket{0}^{\otimes 2^l}$. In this case, the final state after the circuit is also $\ket{0}^{\otimes 2^l}$, since the controlled Hadamard gates are only activated when the first qubit is in the $\ket{1}$ state.

Next, consider the input state $\ket{1}\ket{0}^{\otimes 2^l-1}$. For this input, the circuit should output the first Dicke state. We prove this by induction.
From the circuit in \cref{fig:dicke_state_k_2}:
\begin{itemize}
    \item For $l=0$, $U\ket{1} = \ket{1} = \ket{D^{(1)}_1}$.
    \item For $l=1$, $U\ket{10} = \frac{1}{\sqrt{2}} \left(\ket{10} + \ket{01}\right) = \ket{D^{(2)}_1}$.
\end{itemize}
Thus, for $l=0$ and $l=1$, the circuit generates the correct Dicke state.

\textbf{Inductive Step:}

Assume, as the induction hypothesis, that the circuit in \cref{alg:circuit_dicke_state} creates the first Dicke state from the input $\ket{1}\ket{0}^{\otimes 2^l-1}$, given by
\begin{align}
    \ket{D^{(2^l)}_1} = \frac{1}{\sqrt{2^l}} \sum_{i=1}^{2^l} \ket{0}^{\otimes i-1} \ket{1}_i \ket{0}^{\otimes 2^l-i}.
\end{align}

Now, consider the state $\ket{\psi}$ when the total number of qubits is $2^{l+1}$, after including the controlled Hadamard but before the action of the additional Toffoli and CNOT gates in \cref{alg:circuit_dicke_state}:
\begin{align}
    \ket{\psi} = \frac{1}{\sqrt{2^{l+1}}} \sum_{i=1}^{2^l} \ket{0}^{\otimes i-1} \ket{1}_i \ket{0}^{\otimes 2^l-i} \ket{0}^{\otimes 2^l-1} \left( \ket{0}_{2^{l+1}} + \ket{1}_{2^{l+1}} \right).
\end{align}

The next step is to apply the gate
\begin{align}
    V_{i,j,k} = \mathrm{Toffoli}_{ijk} \mathrm{CX}_{ki} \mathrm{CX}_{kj},
\end{align}
whose action is
\begin{align}
    V\ket{a,b,c} = \ket{a \oplus b,\, a \oplus b,\, c \oplus (a \land b)}.
\end{align}

Following \cref{alg:circuit_dicke_state}, the new state becomes
\begin{align}
    &\frac{1}{\sqrt{2^{l+1}}} \prod_{m=1}^{2^l-1} V_{m,2^{l+1},2^l+m} \sum_{i=1}^{2^l} \ket{0}^{\otimes i-1} \ket{1}_i \ket{0}^{\otimes 2^l-i} \ket{0}^{\otimes 2^l-1} \left( \ket{0}_{2^{l+1}} + \ket{1}_{2^{l+1}} \right) \\
    &= \frac{1}{\sqrt{2^{l+1}}} \left( \sum_{i=1}^{2^l} \ket{0}^{\otimes i-1} \ket{1}_i \ket{0}^{\otimes 2^{l+1}-i} + \sum_{i=1}^{2^l-1} \ket{0}^{\otimes 2^l+i-1} \ket{1}_{2^l+i} \ket{0}^{\otimes 2^{l+1}-2^l-i} + \ket{0}^{\otimes 2^l-1} \ket{1}_{2^l} \ket{0}^{\otimes 2^l-1} \ket{1}_{2^{l+1}} \right).
\end{align}

The action of the final CNOT in \cref{alg:circuit_dicke_state} yields
\begin{align}
    \frac{1}{\sqrt{2^{l+1}}} \left( \sum_{i=1}^{2^l} \ket{0}^{\otimes i-1} \ket{1}_i \ket{0}^{\otimes 2^{l+1}-i} + \sum_{i=1}^{2^l} \ket{0}^{\otimes 2^l+i-1} \ket{1}_{2^l+i} \ket{0}^{\otimes 2^{l+1}-2^l-i} \right),
\end{align}
which can be rewritten as
\begin{align}
    \frac{1}{\sqrt{2^{l+1}}} \sum_{i=1}^{2^{l+1}} \ket{0}^{\otimes i-1} \ket{1}_i \ket{0}^{\otimes 2^{l+1}-i} = \ket{D^{(2^{l+1})}_1}.
\end{align}

Thus, if \cref{alg:circuit_dicke_state} generates the correct state for $k=2^l$ qubits, it also generates the correct state for $k=2^{l+1}$ qubits. By induction, this proves that the algorithm generates the correct Dicke state for all $k=2^l$.

\textbf{Linearity:}

Next, we prove that the circuit in \cref{alg:circuit_dicke_state} satisfies the third condition for Dicke state preparation.

Consider the subspace $P_S$, which is spanned by the states $\{\ket{0}^k, \ket{e}^k\}$, where $\ket{e}^k$ denotes the computational basis states with a single excitation. The action of the unitary $U^{\dagger}$ maps this to the subspace $P_T$, which is spanned by the states $\{\ket{0}^k, \ket{D^{(k)}_1}\}$.
Define the operator
\begin{align}
    M \equiv \mathds{1}_{1} \otimes \left( \mathds{1}_{k-1} - \ketbra{0}{0}^{\otimes k-1} \right) = \mathds{1} - P_S,
\end{align}
where $P_S$ is the projector onto the subspace $P_S$.
It follows that
\begin{align}
    U^{\dagger} M U = \mathds{1} - P_T,
\end{align}
where $P_T$ is the projector onto the subspace $P_T$.
Now, consider the state $\ket{\psi} = U\ket{\phi}$. For $\ket{\psi}$ to have a non-zero overlap with the state $\ket{0}^{k-1}$, we require
\begin{align}
    \bra{\psi} M \ket{\psi} = 0,
\end{align}
which implies
\begin{align}
    1 - \bra{\psi} P_S \ket{\psi} = 0.
\end{align}
Therefore, $\ket{\psi} \in P_S$, or equivalently, $\ket{\phi} \in P_T$. This is only possible if $\ket{\phi}$ is a linear combination of $\ket{0}^k$ and $\ket{D^{(k)}_1}$.

Thus, if the state $\ket{\phi}$ has any contribution from Dicke states with excitation number greater than $1$, the resulting state $\ket{\psi}$ will not have support on $\ket{0}^{k-1}$ under the action of the unitary in \cref{alg:circuit_dicke_state}.

\end{proof}
\textbf{Resource estimation:}

Since each controlled Hadamard gate requires $2$ T gates, 
the circuit in \cref{alg:circuit_dicke_state} requires a total number of resource states given by
\begin{align}
   N_{\mathrm{resource}} &= 2l + \sum_{i=2}^{l} \left(2^{i-1}-1\right) \\
   &= 2l + (2^l - l - 1) \\
   &= 2^l + l - 1,
\end{align}
for a system of $2^l$ qubits, which includes both the Toffoli and T count for the Dicke state preparation circuit.

\subsection{Block-encoding oracles}
\label{app:oracles}

In practice, we run a slightly compressed version in which two oracles are combined. 
This is illustrated in \cref{fig:oracle-details}.
\begin{figure}
    \centering
    \includegraphics[width=1.1\linewidth]{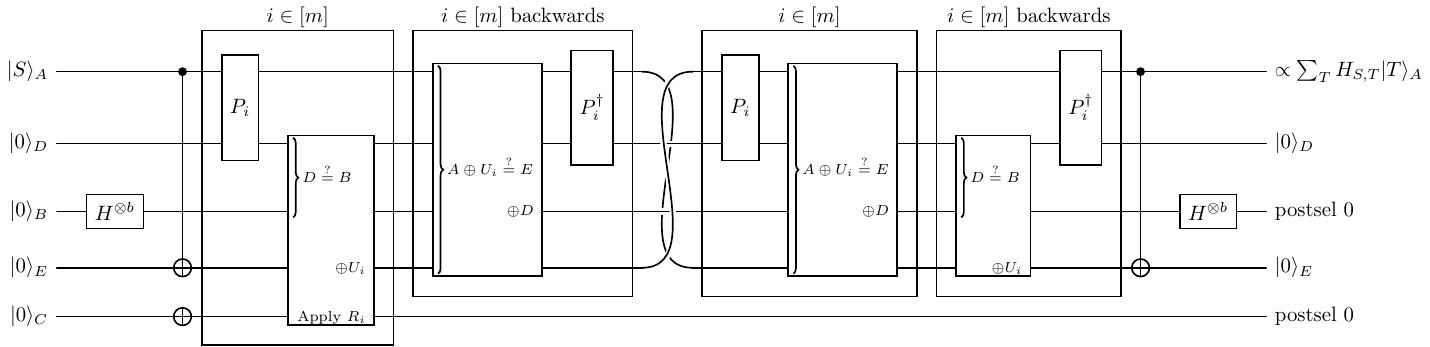}
    \caption{\textbf{The circuit description of $O_H$}. 
    The circuit gives the full oracle $O_H$ by combining the two oracles $O_A$ and $O_E$.
    }
    \label{fig:oracle-details}
\end{figure}        

$P_i$ maps $|S, x\> \mapsto |S, x+1\>$ if $|S+U_i| = \ell$ for a $U_i \in \SC$ and otherwise does nothing. For a given $U_i$ the circuit is shown in Fig.~\ref{fig:circuit_P_i}.
\begin{figure}
    \centering
    \includegraphics[width=0.75\linewidth]{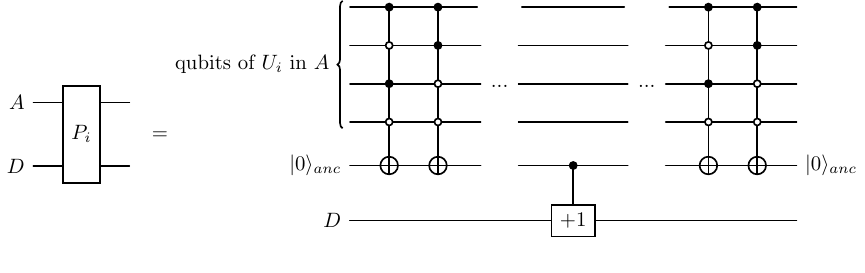}
    \caption{ \textbf{Circuit description of $P_i$.}
 The circuit decomposes each $P_i$ into Toffoli and Incrementor gates, with the Incrementor itself further decomposable into Toffoli gates. For the qubits of $U_i$ in register $A$, we need to apply all possible 4-body Toffoli gates; there are six such gates, and each requires two Toffoli gates to implement. The Incrementor acts on seven qubits, which requires a total of 21 Toffoli gates for its implementation. Therefore, the total number of Toffoli gates needed is $21 + 2 \times 2 \times 6 = 45$.
The Toffoli depth of the Incrementor is $D$, and the depth for the Toffoli gates acting between $A$ and the ancilla ($anc$) can be achieved with a depth of $k(k-1)/2 + 1$. This yields a total circuit depth of $D + 7 = 13$.
    }
    \label{fig:circuit_P_i}
\end{figure}
This is based on the fact that the condition $|S+U_i| = \ell$ is equivalent to $S$ having, when restricted to the $k$ qubits corresponding to $U_i$, exactly $k/2$ ones and $k/2$ zeros. Therefore one can check for all possible ${k \choose k/2}$ such possibilities, and store the result in an ancilla, which is then used to control an increment on the $D$-qubit counter register and then is uncomputed.

For each $U_i$ the counter register is then compared with the stored value of $k$. In $O_E$, when they match the ancilla (initialized with value $|1\>$) is rotated by $R_i = R_y(\arcsin(H_{T(k, S), S}))$ to $H_{T(k, S), S}|0\> + \sqrt{1 - H^2_{T(k, S), S}}|1\>$, and if there is no match ($k > \sigma(S)$) the ancilla is left at $|1\>$. Note that $H_{T(k, S),S} = b_{T(k, S)\Delta S} = b_{i}$ only depends on the current $U_i$. 
See Fig.~\ref{fig:U_1} for more details.
\begin{figure}
    \centering
    \includegraphics[width=0.75\linewidth]{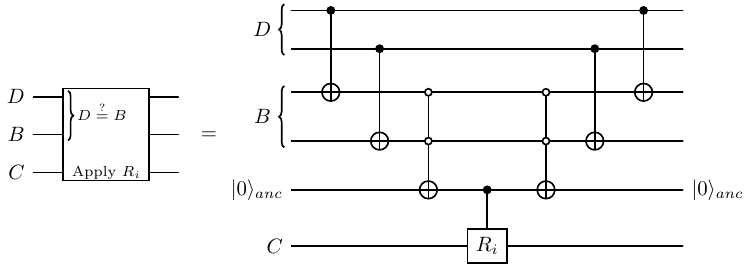}
    \caption{For the above circuit which is a part of the oracle $O_E$. 
    Using the Temporary AND Compute and Uncompute (TACU) gadget in \cite{litinski2022active}, we only need the Toffoli for the uncomputation part.
    Thus the circuit has a Toffoli depth of $B=D=6$.
     The Toffoli acting between the B and ancilla  can be done with a depth of $\lceil \log_2(B)\rceil=3$.
     }
    \label{fig:U_1}
\end{figure}
\begin{figure}
    \centering
    \includegraphics[width=0.75\linewidth]{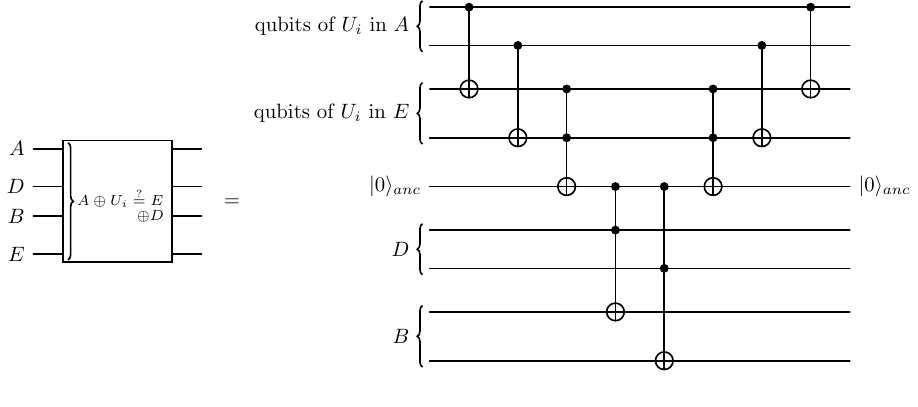}
    \caption{For the above circuit, which is part of the oracle $O_A$, a total of $8$ Toffoli gates are required for its implementation.
    By using $6$ ancilla qubits, the Toffoli depth can be reduced to $3$. Assuming $E = 1$, a depth of $2$ is needed for the Toffoli connections from $E$ to the ancilla. The Toffoli gates between the ancilla, $D$, and $B$ can be performed in a single Toffoli depth if the ancilla information is copied. Thus, the overall Toffoli depth is $3$.}
    \label{fig:U_2}
\end{figure}

In $O_A$, first the value of $|S\>$ is copied in an ancillary register using $n$ CNOTs. Then a similar module to before is used, composed of $P_i$ followed by a comparison-controlled gate, repeated for all $i \in [m]$. The difference is that the controlled gate now performs an XOR on the ancillary register with the qubits corresponding to $U_i$, thus producing $T(k,S)$ whenever the counter matches with $k$. 
Then, another sequence of repeated gates is performed, for each $i \in [m]$ backwards. First a gate compared the value of the first register (containing $|S\>$) with the value of the last (containing $|S\Delta U_j\>$ for some $j$, or $|S\>$), modulo an XOR by $U_i$. This can be done efficiently with the following circuit:
\begin{figure}[H]
    \centering
    \includegraphics[width=0.65\linewidth]{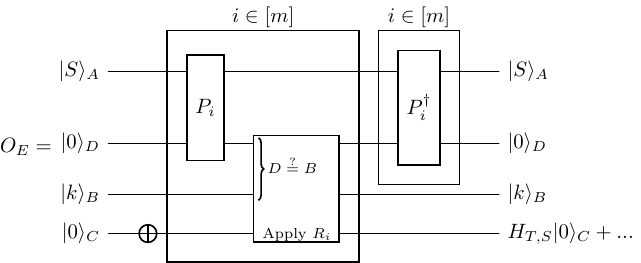}
    \caption{\textbf{Full circuit description of $O_E$}}
    \label{fig:circuit_O_E}
\end{figure}
This works because the first and the last register differ on the $k$ qubits corresponding to a single $U_j$, and therefore it is sufficient to compute the XOR of the qubits corresponding to the current $U_i$ and check if they all equal 1.
Conditional to this, the value of the counter register is added via XOR to the register storing $k$, and finally $P_i^\dagger$ is applied. This uncomputes $k$ since $P_i^\dagger$ decreases the counter until it equals $k$, at which point the comparison succeeds and therefore $k$ is erased. Then the remaining $P_i^\dagger$'s reduce the counter to 0.
\begin{figure}[H]
    \centering
    \includegraphics[width=0.65\linewidth]{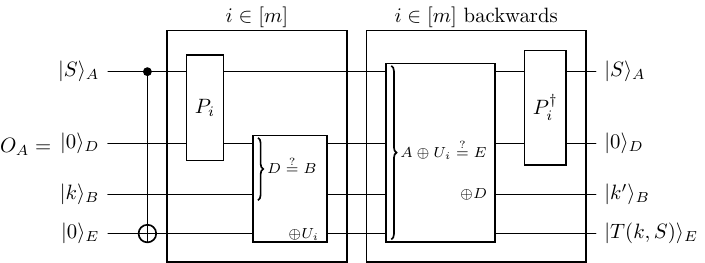}
    \caption{\textbf{Full circuit description of $O_A$}}
    \label{fig:circuit_O_A}
\end{figure}

\subsection{Parallelizing the clauses}
For the tensor PCA problem the circuit depth associated with the clauses scales as $m = n^2$.
In this section, we focus on reducing the depth to $\mathcal{O}(n)$. 

There are two key ingredients in achieving this depth reduction:
\begin{itemize}
    \item \textbf{Parallelization of clauses:} We employ an off-the-shelf algorithm to parallelize the evaluation of clauses.
    \item \textbf{Depth reduction from dense encoding:} We optimize the circuit structure to further reduce the depth arising from the dense encoding of the problem.
\end{itemize}

\textbf{Parallelization of clauses:} We consider a random 4-CNF formula with $n$ variables and $m$ clauses.  
Each clause contains four distinct variables, so the total number of variable
occurrences is $4m$. For a fixed variable $x_i$, the number of occurrences
follows
\[
r_i \sim \mathrm{Binomial}\!\left(4m,\tfrac{1}{n}\right),
\]
with mean $\mathbb{E}[r_i]=4m/n$ and variance
$\mathrm{Var}(r_i)=\tfrac{4m}{n}(1-\tfrac{1}{n})$.  
Hence the typical load is $4m/n$, and the maximum load $r_{\max}$ concentrates
around $4m/n + O(\sqrt{(4m/n)\log n})$.

The numerical simulation of clause parallelization using the graph coloring algorithm is illustrated in \cref{fig:graph_coloring}. In this simulation, we generate a random graph where each node represents a variable, and edges correspond to clauses that connect variables. Specifically, we consider $n$ variables and construct $m = 10n^2 \log(n)$ clauses, with each clause involving $4$ distinct variables.
The graph coloring algorithm is then applied to assign colors to the nodes such that no two adjacent nodes (i.e., variables appearing in the same clause) share the same color. This coloring corresponds to scheduling clauses in parallel without conflicts. The asymptotic bound for the number of colors required by the graph coloring algorithm in this scenario is given by $4m/n$. 
This bound reflects the maximum degree of the graph, which is determined by the number of clauses each variable participates in. 
As $n$ increases, the bound provides insight into the scalability and efficiency of parallelization achievable through graph coloring.

\textbf{Depth reduction from dense encoding:} 
In the first step of the state preparation, we use a dense encoding where we create  equal superposition of  states in $s=\log_2(m)$ qubits. 
Then we need to add a phase according to the specific bit string which will require a depth of $m\times \log_2(s)$.
Here we reduce the depth by trading off space, 
to achieve a depth reduction, first we make $\lceil n/4 \rceil$  copies of the state which requires a CNOT dpeth of $\log_2(\lceil n/4 \rceil)$.
Then we write each of the possible parallel clauses to the $\lceil 5n/2 \rceil$ ancilla followed by a phase operation according to the clause.
Then for each clause we map to the sparse states which is encoded in the data qubits of size $n$ which requires a $k$ CNOT gates.
Thus reducing the depth to $4 m/n\log_2(s)$ by using a total number of $\lceil n/4\rceil \times s$  extra qubits.

To ensure a complete depth reduction, we also need to reduce the depth in the uncomputation part which also accompanies the state preperation. 
In the uncomputation step, we apply a multicontrolled Toffoli gates on the qubits in the dense encoding for bitstrings correspond to $1$ corresponding the for the state in the $n$ qubits.
To achieve this we apply a $k+s$ body multicontrolled Toffoli gate. However, again the depth of this is $m \log_2(k+s)$.
To reduce this depth, one can first write the $k$- body condition to an ancilla which can be done in parallel for $\lceil 5n/2\rceil$ clauses.
Further we use additional $\lceil n/4\rceil \times s$ ancilla and CNOT gates such that one can implement this step with a non-clifford depth of $\sim \lceil 4m/n \rceil \times \log_2(k)$.

\begin{figure}
    \centering
    \includegraphics[width=0.5\linewidth]{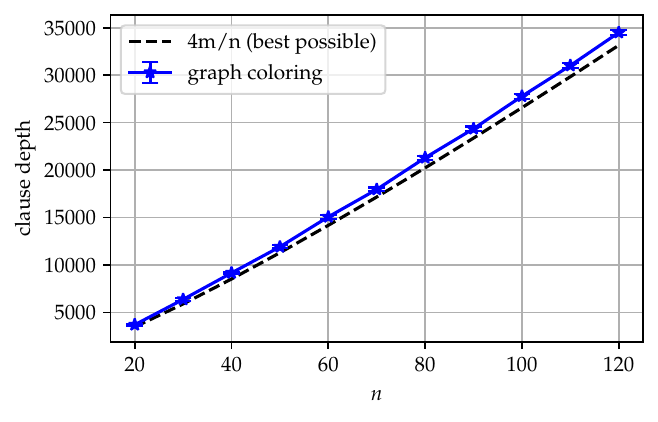}
    \caption{\textbf{Parallelizing the clauses.} 
    The depth of the clause 
empirically found  using the graph coloring algorithm.
For graph coloring we construct a random graph for $n$ variables with $m=10n^2\log n$ clauses and each clause addresses $4$ variables.
The asymptotic bound for the graph coloring algorithm is given as $4m/n$. }
    \label{fig:graph_coloring}
\end{figure}

\section{QSP rotation accuracy}

Quantum Signal Processing (QSP) is a powerful framework for implementing polynomial transformations of quantum operators using sequences of controlled rotations.
A critical aspect of QSP is the precise determination of rotation angles, which directly influence the accuracy of the implemented transformation.
In practical scenarios, these angles must be decomposed into a finite set of elementary operations, introducing a \textit{decomposition error}. Understanding and quantifying this error is essential for assessing the fidelity of QSP-based algorithms and for guiding the design of efficient quantum circuits. 
This section explores the sources and implications of decomposition error in QSP for \textit{Tensor algorithms}, focusing on how the approximation of rotation angles affects overall algorithmic performance.

Given a function $f$ that one wishes to apply to the spectrum of a unitary operator $U$, the first step is to classically generate a (Laurent) polynomial $P$ that closely approximates $f$ over the desired spectral range. 
Once a suitable polynomial approximation is obtained, one can employ efficient classical algorithms to compute the sequence of single-qubit rotation angles—known as the QSP phases—which are interspersed with applications of the controlled unitary (the QSP signal) and correspond to the polynomial approximation. These rotation angles must then be decomposed into a sequence of Clifford and $T$ gates, introducing a fixed decomposition error $\epsilon$. 
The overall operator norm of the implemented transformation is subsequently evaluated, taking into account the error introduced by the decomposition. 
By appropriately choosing the value of $\epsilon$, one can ensure that the overall operator norm of the implementation remains below the total number of repetitions of the QSP protocol, thereby maintaining the desired accuracy of the quantum algorithm.

\begin{figure}[htbp]
    \subfloat[][Exact QSP]{%
        \includegraphics[width=0.32\textwidth]{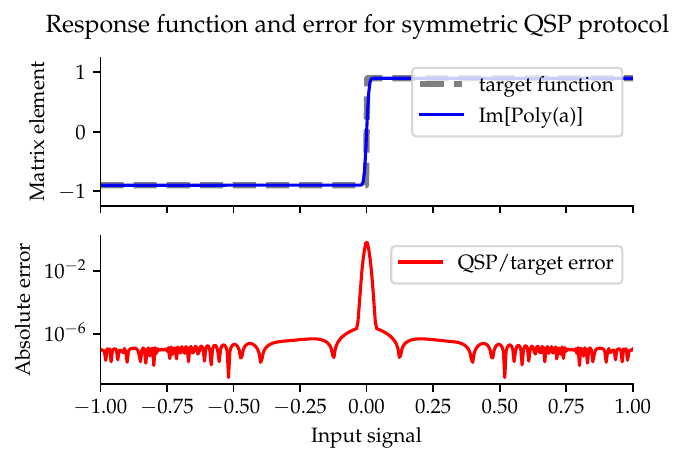}%
        \label{fig:infinite_dimensional_analysis_1}%
    }
    \hfill
    \subfloat[][QSP for $\epsilon=1e-8$]{%
        \includegraphics[width=0.32\textwidth]{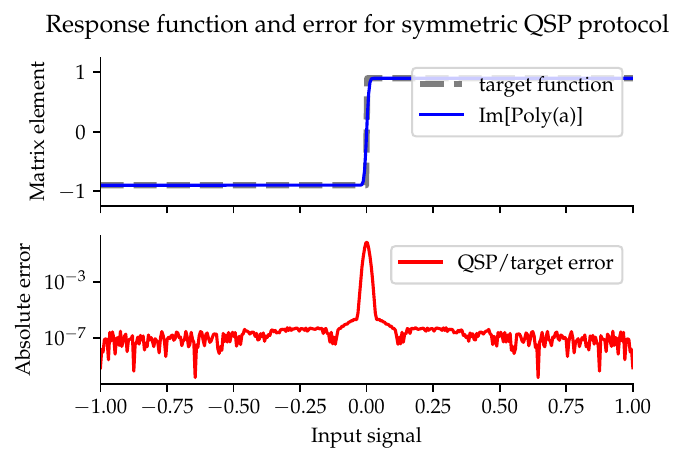}%
        \label{fig:infinite_dimensional_analysis_2}%
    }
    \hfill
    \subfloat[][Operator norm]{%
        \includegraphics[width=0.34\textwidth]{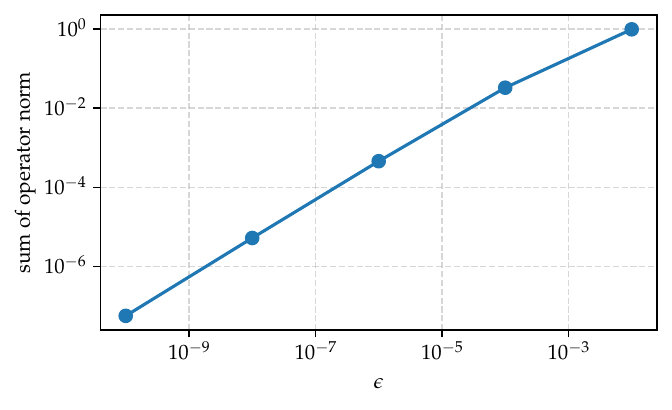}%
        \label{fig:infinite_dimensional_analysis_12}%
    }
  \caption{
        \textbf{QSP Response Function and Decomposition Error Analysis.}
       (a) The exact Quantum Signal Processing (QSP) response function for a target polynomial transformation, showing the imaginary part of the polynomial as a function of the input signal. (b) The QSP response function with a finite decomposition error $\epsilon = 10^{-8}$, illustrating the impact of angle discretization on the transformation accuracy. (c) The sum of operator norms as a function of the decomposition error $\epsilon$, demonstrating how the overall implementation error scales with the chosen precision. These plots collectively highlight the trade-off between decomposition accuracy and resource requirements in QSP-based quantum algorithms.
    }
    \label{fig:qsp_error}
\end{figure}

Figure~\ref{fig:qsp_error} illustrates the effect of rotation angle decomposition error on the performance of Quantum Signal Processing (QSP) protocols. 
Panel (a) shows the exact response function for a symmetric QSP protocol, where the imaginary part of the polynomial closely matches the target function across the input signal range.
Panel (b) presents the response function when the QSP rotation angles are decomposed with a finite error $\epsilon = 10^{-8}$, revealing small but noticeable deviations from the exact transformation. Panel (c) quantifies the cumulative operator norm error as a function of the decomposition error $\epsilon$, indicating that the total implementation error increases with larger $\epsilon$. These results demonstrate the importance of precise angle synthesis in maintaining the fidelity of QSP-based quantum algorithms, and provide guidance for selecting an appropriate decomposition error to balance accuracy and resource cost.

\section{Resource estimation of individual components}
Given the circuits one can estimate the resource requirements, both the non-Clifford gate count $\mathcal{N}$ and depth $\mathcal{D}$ required to implement each individual component.

\textbf{Resource estimation for the Hamming weight preserving circuit in \cref{subsec:guiding_state_preperation}:}

Since each controlled Hadamard gate requires $2$ T gates, 
the circuit in \cref{alg:circuit_dicke_state} requires a total number of resource states given by
\begin{align}
   \mathcal{N}\left(\text{guiding state}\right) &= 2l + \sum_{i=2}^{l} \left(2^{i-1}-1\right) \\
   &= 2l + (2^l - l - 1) \\
   &= 2^l + l - 1,
\end{align}
 where $c=2^l$, which includes both the Toffoli and T count for the Dicke state preparation circuit.
 The depth of the guiding state preperation is,
 \begin{align}
     \mathcal{D}(\text{guiding state})=2+(2^l+l-1)=2^l-l+1.
 \end{align}

\textbf{Resource estimation for the dense to sparse encoding }

The non-Clifford count and depth of the dense to sparse encoding is given as,
\begin{align}
    \mathcal{N}(\text{sparse})&=c\times m \times (s+k)\\
    \mathcal{D}(\text{sparse})&=4\times n \times (\log_2(s)+\log_2(k))
\end{align}

\textbf{Resource estimation for $O_H$: }

As described in circuit \cref{fig:oracle-details}, the 
non-Clifford count and depth of the dense to sparse encoding is given as,
\begin{align}
    \mathcal{N}(O_H)&=c\times m \times b\\
    \mathcal{D}(O_H)&=4\times 4n \times b'
\end{align}
where $b$ and $b'$ are the cost and depth of implementing each term in the algorithm respectively.
We found that that for the case of $k=4$, we get $b=210$ and $b'=60$. 
\begin{table*}
\centering
 \begin{tabular}{|c|c|c|c|c|c|c|}
 \hline
 component
    & gate count 
    & depth 
    &  qubits
    \\ 
    \hline
    Dicke state
    &5 
    &3
    &4\\
    \hline
    Dense to Sparse encoding
    &$c\times m\times (s+k)$
    & $2 m\times \left(\log_2(s)+\log_2(k)\right)$  
    & $c\times n+\lceil n/4\rceil \left(s+1\right)$\\
    \hline
        $O_H$ 
        & $\sim m 
        \times b$  
        & $\sim 40 n \times b'$
        & $n$ \\
    \hline
  \end{tabular}
    \caption{\textbf{Resource estimation for the individual components for the algorithm.}
    The gate count and the depth of implementing one term in $O_H$ is given as $b$ and $b'$ respectively.
    For the case of $k=4$ and $c=4$, we get, $b=210$ and $b'=60$ respectively.}
    \label{tab:resource_estimation_inividual_components}
\end{table*}

\begin{table}
\centering
 \begin{tabular}{|c|c|c|c|c|c|}
 \hline
 component
    &cost

    \\ 
    \hline
    $\mathcal{N}\left({\Pi_\psi}\right)$ 
    & $2\times c^{l/2}\left[cm(k+s)+(10m+2c(n-1))\right]+n\log_2(1/\epsilon)  $  
    \\
    \hline
    $\mathcal{N}\left({\Pi_\phi}\right)$
    &$q\times \left[4 m\times b +7n-2+3 \log_2(1/\epsilon)\right]$
 
    \\
    \hline
    $\mathcal{D}\left({\Pi_\psi}\right)$ 
    &$2\times c^{l/2} \left[4\frac{m}{n}\times (\log_2(k)+\log_2(s))+  \left(24n +2\log_2(c(n-1))\right)\right]+\log_2(1/\epsilon)$
 
    \\
    \hline
         $\mathcal{D}\left({\Pi_\phi}\right)$
        &$q\times \left[4\times \frac{m}{n}\times b'+3\log_2(n-1)+2+3\log_2(1/\epsilon)\right]$
       
         \\
    \hline
  \end{tabular}
    \caption{\textbf{Resource estimation for the two components of the floating point amplitude amplification.}  
 $\mathcal{N}$ and $\mathcal{D}$ corresponds to the gate count and depth respectively.
 $\epsilon$ is the decomposition accuracy and through the paper we fix the rotation accuracy to be $\epsilon=10^{-10}$.
    }
    \label{tab:resource_estimation_overall}
\end{table}

\end{document}